\documentclass{lmcs} 
\pdfoutput=1

\usepackage{lastpage}
\lmcsdoi{18}{1}{42}
\lmcsheading{}{\pageref{LastPage}}{}{}%
{Jan.~12,~2021}{Mar.~23,~2022}{}

\keywords{register automata, two-variable logic, decidability}

\usepackage{subcaption}
\usepackage{subcaption}
\usepackage{thmtools} 
\usepackage{thm-restate}

\usepackage{amsmath, amsfonts, amsthm,amssymb}
\usepackage{rotating,fancyhdr,graphicx}

\usepackage{fixltx2e,ifthen, sidecap, wrapfig, xspace}

\usepackage{mathrsfs, stmaryrd}

\usepackage{
 calc, colortbl
}

\usepackage[font=small,format=plain,labelfont=sc,textfont=it]{caption}
\usepackage[T1]{fontenc}
\usepackage[utf8]{inputenc}
\usepackage[all]{xy}
\usepackage[geometry]{ifsym}
\usepackage{todonotes}
\usepackage{tikz}
\usetikzlibrary{backgrounds,calc,trees}

\usetikzlibrary{arrows,automata,decorations.pathreplacing,decorations.pathmorphing,decorations.footprints,fadings,calc,trees,mindmap,shadows,decorations.text,patterns,positioning,shapes,matrix,fit}

\usepackage{multicol} \usetikzlibrary{decorations.pathmorphing}

\newcommand{\wh}{\widehat}

\newcommand{\sz}[1]{}
\newcommand\muparrow{\mathord{\scriptstyle \uparrow}\hspace{0.7mm}}
\newcommand\mdownarrow{\mathord{\scriptstyle \downarrow}\hspace{0.7mm}}
\newcommand\msearrow{\mathord{\scriptstyle \searrow}\hspace{0.7mm}}
\newcommand\mswarrow{\mathord{\scriptstyle \swarrow}\hspace{0.7mm}}

\newcommand{\rw}[1]{r_{#1}}

\newcommand{\rsup}{r_{\sup}}

\newcommand{\rinf}{r_{\inf}}

\newcommand{\whA}{\wh{\cal A}}

\newcommand{\types}{\textrm{Types}}
\newcommand{\type}{\ensuremath{\textsc{tp}}}

\let\olduparrow\uparrow
\renewcommand{\uparrow}{\olduparrow \hspace{-2pt}}
\let\oldsearrow\searrow
\renewcommand{\searrow}{\oldsearrow \hspace{-2pt}}
\let\oldswarrow\swarrow
\renewcommand{\swarrow}{\oldswarrow \hspace{-3pt}}
\renewcommand{\subset}{\subseteq}

\theoremstyle{plain}

\newenvironment{proofof}[1]{\begin{proof}[Proof of~#1]}{\end{proof}}

\newtheorem{theorem}[thm]{Theorem}
\newtheorem{corollary}[thm]{Corollary}
\newtheorem{example}[thm]{Example}
\newtheorem{lemma}[thm]{Lemma}
\newtheorem{proposition}[thm]{Proposition}

\newcommand{\fo}{\ensuremath{\mathrm{FO}}\xspace}
\newcommand{\fotwo}{\ensuremath{\mathrm{FO}^2}\xspace}

\newcommand{\mso}{\ensuremath{\mathrm{MSO}}\xspace}

\newcommand{\emsotwo}{\ensuremath{\mathrm{EMSO}^2}\xspace}
\newcommand{\esotwo}{\ensuremath{\mathrm{ESO}^2}\xspace}

\newcommand{\R}{\mathbb{R}}
\newcommand{\Q}{\mathbb{Q}}

\newcommand{\N}{\mathbb{N}}

\newcommand{\Rbar}{\overline{\R}}

\newcommand{\df}{\stackrel{\mathit{df}}{=}}
\newcommand\set[1]{\ensuremath{\{#1\}}}
\newcommand{\setof}[2]{\set{#1\mid#2}}

\newcommand{\from}{\colon}
\newcommand{\func}[3]{\mathop{\mathchoice{{#1}\quad\colon\quad{#2}\quad\longrightarrow\quad {#3}}{{#1}\colon{#2}\to{#3}}{script}{sscript}
  }}

\newcommand{\data}{\mathbf D}

\newcommand {\frakA} {\ensuremath{\mathfrak{A}}}

\newcommand {\calA}      {{\cal A}\xspace}

\newcommand {\calB}      {{\cal B}\xspace}
\newcommand {\calC}      {{\cal C}\xspace}

\newcommand {\calK}      {{\cal K}\xspace}

\newcommand {\calP}      {{\cal P}\xspace}
\newcommand {\calR}      {{\cal R}\xspace}
\newcommand {\calS}      {{\cal S}\xspace}
\newcommand {\calT}      {{\cal T}\xspace}

\newcommand {\calL}      {{\cal L}\xspace}

\newcommand{\aut}[1]{\cal{#1}}

\newcommand{\cal}[1]{\mathcal{#1}}

\tikzstyle{mnode}=[
  circle,
  fill=\mnodefillcolor, 
  draw=\mnodedrawcolor,
  minimum size=6pt, 
  inner sep=0pt
]
\tikzstyle{blueNode}=[
  circle,
  fill=blue!50, 
  draw=black!90,
  minimum size=6pt, 
  inner sep=0pt
]
\tikzstyle{redNode}=[
  circle,
  fill=red!50, 
  draw=black!90,
  minimum size=6pt, 
  inner sep=0pt
]
\tikzstyle{greenNode}=[
  circle,
  fill=green!70!black, 
  draw=black!90,
  minimum size=6pt, 
  inner sep=0pt
]

\tikzstyle{mnodeinvisible}=[
  minimum size=6pt, 
  inner sep=0pt
]

\tikzstyle{invisible}=[
  minimum size=0pt, 
  inner sep=0pt
]

\tikzstyle{invisiblel}=[
  minimum size=10pt, 
  inner sep=0pt
]

\tikzstyle{invisibleEdge}=[
  transparent
]

\tikzstyle{nameNode}=[
  font=\scriptsize
]

\tikzstyle{namingNode}=[
  font=\normalsize
]

\tikzstyle{mEdge}=[
  -latex', thick, 
  shorten >=3pt, 
  shorten <=3pt,
  draw=black!80,
]

\tikzstyle{dDashedEdge}=[
  -latex', thick, 
  shorten >=3pt, 
  shorten <=3pt,
  draw=black!80,
  dashed
]

\tikzstyle{dEdge}=[
  -latex', thick, 
  shorten >=3pt, 
  shorten <=3pt,
  draw=black!80,
]

\tikzstyle{dhEdge}=[
  -latex', thick, 
  shorten >=3pt, 
  shorten <=3pt,
  draw=black!80,
]
\tikzstyle{uEdge}=[
thick, 
  shorten >=3pt, 
  shorten <=3pt,
  draw=black!80,
]

\tikzstyle{blueEdge}=[
thick, 
  shorten >=1pt, 
  shorten <=1pt,
  draw=blue!80,
]
\tikzstyle{redEdge}=[
thick, 
  shorten >=1pt, 
  shorten <=1pt,
  draw=red!60,
]
\tikzstyle{greenEdge}=[
thick, 
  shorten >=1pt, 
  shorten <=1pt,
  draw=green!70!black,
]

\tikzstyle{uhEdge}=[
thick, 
  shorten >=3pt, 
  shorten <=3pt,
  draw=black!80,
]

\tikzstyle{cEdge}=[
ultra thick, 
  shorten >=3pt, 
  shorten <=3pt,
  draw=black!80,
]

\tikzstyle{dotsEdge}=[
  thick, 
  loosely dotted, 
  shorten >=7pt, 
  shorten <=7pt
]

\tikzstyle{snakeEdge}=[
  ->, 
  decorate, 
  decoration={snake,amplitude=.4mm,segment length=2.5mm,post length=0.5mm},
]

\tikzstyle{snakeEdgea}=[
  ->, 
  decorate, 
  decoration={snake,amplitude=.4mm,segment length=3mm,post length=0.5mm}
]

\tikzstyle{class rectangle}=[
  draw=black,
  inner sep=0.2cm,
  rounded corners=5pt,
  thick
]

\tikzstyle{mline}=[
  draw=black,
  inner sep=0.2cm,
  rounded corners=5pt,
  thick
]

\tikzstyle{mainclass rectangle}=[
draw=blue,
  inner sep=0.2cm,
  rounded corners=5pt,
  very thick
]

\newcommand{\mnodedrawcolor}{black!80}
\newcommand{\mnodefillcolor}{black!40}

\pgfdeclarelayer{background}
\pgfdeclarelayer{substructure}
\pgfdeclarelayer{edges}
\pgfdeclarelayer{foreground}
\pgfsetlayers{background,substructure,edges,main,foreground}

\tikzstyle{background rectangle}=[
  fill=black!5,
  draw=black!20,
  inner sep=0.2cm,
  rounded corners=5pt
]

\begin{document}

\title[Register Automata with Extrema Constraints]{Register Automata with
  Extrema Constraints,\texorpdfstring{\\}{ }and an Application to Two-Variable Logic}

\titlecomment{The short version of this article appeared in the conference proceedings of LICS 2020.}

\author{Szymon Toru\'nczyk}[a]	\address{University of Warsaw}	\email{szymtor@mimuw.edu.pl}

\thanks{%
The work of S.~Toru\'nczyk was supported by Poland's National Science Centre grant 2018/30/E/ST6/00042.
The work of T.~Zeume was supported by the Deutsche Forschungsgemeinschaft
(DFG, German Research Foundation), grant 448468041.}

\author{Thomas Zeume}[b]	\address{Ruhr University Bochum}	\email{thomas.zeume@rub.de}

\begin{abstract}
	We introduce a model of register automata
over infinite trees with extrema constraints. Such an automaton  can store elements of a linearly ordered domain in its registers, and
can compare those values to the suprema and infima of register values in subtrees. We show that the emptiness problem for these automata is decidable. 

As an application, we prove decidability of the countable satisfiability problem for two-variable logic in the presence of a tree order, a linear order, and arbitrary atoms that are MSO definable from the tree order. As a consequence, the satisfiability problem for two-variable logic with arbitrary predicates, two of them interpreted by linear orders, is decidable.  \end{abstract}

\maketitle

\section{Introduction}\label{section:introduction}
Automata for words and trees find applications in diverse areas such as logic, verification, and database theory (see, e.g., \cite{Thomas97, Baier08, Neven02}). Applications to logic include proofs of decidability of the satisfiability problem for various logics, and this is the theme of this paper.
Many variations of automata for specific applications have been introduced, among them automata over infinite words or trees, with output, timed automata, or automata working on words and trees whose positions are annotated by data from an infinite domain. In this article we study a variant of the latter family of automata called \emph{register automata}~\mbox{\cite{KaminskiF94, NevenSV04}}.

In its basic form, a register automaton extends a finite state automaton by registers which can store data values from an infinite domain $\data$. The inputs are \emph{data words}, i.e., words labeled by pairs consisting of a label from a finite alphabet $\Sigma$, and a data value from $\data$. When reading a data word, a register automaton can store values from~$\data$ in its registers. Its state depends on the previous state, the label at the current position as well as the relationship of the stored register values to the data value at the current position. Here, depending on the automaton model at hand, register values can be tested for equality, compared with respect to some linear order on $\data$, or others. 

In this article we study a variant of register automata for infinite \emph{data trees}, where the data values form a complete dense total order.
 In addition to the ability of comparing data values according to the linear order of $\data$, our automaton model allows to compare register values to the suprema and infima over values of registers in a subtree. 
 
We show that the emptiness problem for this automaton model can be solved algorithmically. 

\begin{theorem}\label{theorem:emptiness}
  The emptiness problem for tree register automata with suprema and infima constraints is decidable. 
\end{theorem}

As an application of the above result, we consider the satisfiability problem for variants of two-variable logic. In two-variable first-order logic (short: \fotwo) formulas may use only two variables $x$ and $y$ which can be reused. The extension by existential second-order quantifiers is denoted by \esotwo, or \emsotwo if only monadic such quantifiers are allowed. Two-variable first-order logic is reasonably expressive and enjoys a decidable satisfiability problem \cite{Mortimer75, Scott1962}. However, an easy application of a two-pebble Ehrenfeucht-Fra\"{\i}sse game yields that $\fotwo$ cannot express transitivity of a binary relation. For this reason, $\fotwo$ has been studied on structures where some special relations are required to be transitive.

A particular interest has been in deciding the satisfiability problem for such extensions. Recently a decision procedure for the finite satisfiability problem for $\esotwo$ with one transitive relation and for $\esotwo$ with one partial order have been obtained \cite{Pratt-Hartmann18}. Previously $\esotwo$ with two equivalence relations \cite{KieronskiO12, KieronskiT09} and $\esotwo$ with two ``forest'' relations have been shown to be decidable \cite{CharatonikW13}. While it is known that $\emsotwo$ with three equivalence relations is undecidable, this problem is still open for three ``forest'' relations. For $\esotwo$ with two linear orders, only a decision procedure for the finite satisfiability problem was known \cite{SchwentickZ12, ZeumeH16}. The satisfiability problem and the finite satisfiability problem for $\emsotwo$ is undecidable for three linear orders~\cite{Kieronski11}.

The question whether two-variable logic with two linear orders is decidable for general (not necessarily finite) structures has been left open in \cite{SchwentickZ12, ZeumeH16}, and is settled here affirmatively. Beyond settling the question itself, we believe that the techniques developed here might also be interesting in their own rights and applied to other problems, much like in the case of finite satisfiability, where the used techniques were later exploited in work by Dartois, Filiot, and Lhote on transducers \cite{DartoisFL18} and in recent work by Lutz, Jung, and Zeume in relation to description logics \cite{JungLZ20}.

In fact, we prove a more general result. A partial order $(\data, <)$ is a \emph{tree order} (also called a \emph{semi-linear order}) if the set $\{y \mid y < x\}$ is totally ordered by $<$ for each $x \in  \data$, and any two elements $x,y$ have some lower bound.

\begin{theorem}\label{theorem:satisfiability}
  Countable satisfiability of $\;\esotwo$ with one tree order, one linear order, and access to MSO-defined atoms over the tree order, is decidable. 
\end{theorem}
See Theorem~\ref{theorem:maintheorem} for a more precise restatement.

This theorem can alternatively be viewed from the perspective of $\esotwo$ on data trees where all nodes are annotated by \emph{distinct}, linearly ordered data values. It then states that $\esotwo$ with access to the tree structure via MSO-definable atoms and with the ability to compare data values with respect to the linear order on data values is decidable over such trees. This should be compared with the decidability of $\emsotwo$ on data trees with possibly non-distinct data values, access to the tree structure via the children and sibling relation, as well as the ability to test whether two data values are equal \cite{BojanczykMSS09}.

An immediate consequence of Theorem \ref{theorem:satisfiability} is the decidability of satisfiability of $\esotwo$ with two linear orders, because a linear order can be axiomatised from a tree order easily.  
\begin{corollary}\label{cor:satisfiability}
  Satisfiability of $\;\esotwo$ with two linear orders is decidable. 
\end{corollary}

In Section \ref{section:FO2toRA} these results are stated more formally, and other consequences of Theorem \ref{theorem:satisfiability} are discussed. We now briefly discuss our proof method.

 Theorem~\ref{theorem:satisfiability} is proved by
 a reduction to the emptiness problem for our variant of register automata. To explain this reduction, let us consider the case of finite structures first. The finite satisfiability problem for two-variable logic with two linear orders,  $\fotwo(<_1, <_2)$, can be reduced to an emptiness test for register automata on words in two steps: (1) exhibit a correspondence between structures with two linear orders and input data words for register automata, and (2) verify the conditions imposed by a $\fotwo(<_1, <_2)$-formula with a register automaton. For (1), finite structures with two linear orders $<_1$ and $<_2$ can be identified with finite data words with \emph{distinct} data values from the rationals. Here, the unary types of single elements of the structure are represented as labels of the data word, the order $<_1$ corresponds to the linear order of positions in the word, and $<_2$ to the usual order of the rationals on data values, i.e., if positions $x$ and $y$ have data values $p$ and $q$, respectively,  then $x <_2 y$ if and only if $p < q$ (where $<$ is the usual linear order of the rationals). We refer to Figure~\ref{figure:structureanddatawords} for an illustration of this correspondence.

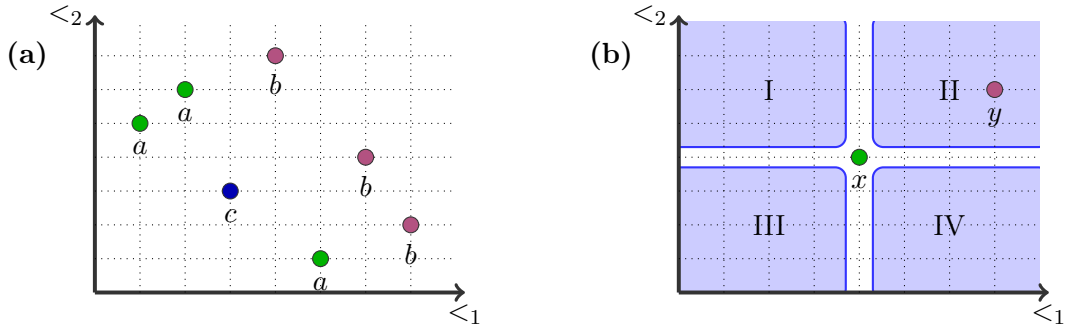
\begin{figure}[t] 
\centering 
    \begin{subfigure}[t]{0.5\textwidth}
      \begin{tikzpicture}[
        xscale=0.6,
        yscale=0.45
      ]
        \node (tmp) at (-1.5, 7) {\textbf{(a)}};
        
\draw [->, line width=1.5pt, black!80] (0,0) -- (8.2,0) node[black, below] {$<_1$};
        \draw [->, line width=1.5pt, black!80] (0,0) -- (0,8.2) node[black, left, sloped] {$<_2$};

\foreach \x in {1,2,...,7}
        \draw[dotted, black] (\x,0) -- (\x,8) ;

        \foreach \y in {1,2,...,7}
        \draw[dotted, black] (0,\y) -- (8,\y);

\node (tmp) at (1, 5)[mnode, label={below:$a$}, fill=black!30!green] {};
        \node (tmp) at (2, 6)[mnode, label={below:$a$}, fill=black!30!green] {};
        \node (tmp) at (3, 3)[mnode, label={below:$c$}, fill=black!30!blue] {};
        \node (tmp) at (4, 7)[mnode, label={below:$b$}, fill=black!30!magenta] {};
        \node (tmp) at (5, 1)[mnode, label={below:$a$}, fill=black!30!green] {};
        \node (tmp) at (6, 4)[mnode, label={below:$b$}, fill=black!30!magenta] {};
        \node (tmp) at (7, 2)[mnode, label={below:$b$}, fill=black!30!magenta] {};
        
\end{tikzpicture}
    \end{subfigure}
    \begin{subfigure}[t]{0.45\textwidth}
      \begin{tikzpicture}[
        xscale=0.6,
        yscale=0.45
      ]
        \node (tmp) at (-1.5, 7) {\textbf{(b)}};
\draw [->, line width=1.5pt, black!80] (0,0) -- (8.2,0) node[black, below] {$<_1$};
        \draw [->, line width=1.5pt, black!80] (0,0) -- (0,8.2) node[black, left, sloped] {$<_2$};

\foreach \x in {1,2,...,7}
        \draw[dotted, black] (\x,0) -- (\x,8) ;

        \foreach \y in {1,2,...,7}
        \draw[dotted, black] (0,\y) -- (8,\y);

\node (tmp) at (4, 4)[mnode, label={below:$x$}, fill=black!30!green] {};
        \node (tmp) at (7, 6)[mnode, label={below:$y$}, fill=black!30!magenta] {};
        \node (tmp) at (2, 2) {III};
        \node (tmp) at (2, 6) {I};
        \node (tmp) at (6, 2) {IV};
        \node (tmp) at (6, 6) {II};

        \begin{pgfonlayer}{background}
          \begin{scope}
            \clip(0,0) rectangle (4,4);
            \fill[draw=blue!80, fill=blue!20, thick, rounded corners] (-0.2,-0.2) rectangle (3.7,3.7);
          \end{scope}
          \begin{scope}
            \clip(8,8) rectangle (4,4);
            \fill[draw=blue!80, fill=blue!20, thick, rounded corners] (8.2,8.2) rectangle (4.3,4.3);
          \end{scope}
          \begin{scope}
            \clip(8,0) rectangle (4,4);
            \fill[draw=blue!80, fill=blue!20, thick, rounded corners] (8.2,-0.2) rectangle (4.3,3.7);
          \end{scope}
          \begin{scope}
            \clip(0,8) rectangle (4.0,4.0);
            \fill[draw=blue!80, fill=blue!20, thick, rounded corners] (-0.2,8.2) rectangle (3.7,4.3);
          \end{scope}

        \end{pgfonlayer}
      \end{tikzpicture}
    \end{subfigure}
        \caption{\textbf{(a)} A structure with two linear orders represented as a point set in the two-dimensional plane. The structure satisfies the existential constraint $\forall x \exists y (a(x) \rightarrow (b(y) \wedge x <_1 y \wedge x <_2 y))$, and the universal constraint $\forall x \forall y \neg(a(x) \wedge b(y) \wedge y <_1 x \wedge y <_2 x)$. A linearly ordered data word corresponding to the structure is $(a, 5)(a, 6)(c, 3)(b, 7)(a, 1)(b, 4)(b, 2)$. \textbf{(b)} In structures with two linear orders, constraints on an element $x$ can be imposed in directions (I)--(IV). For instance, direction (II) corresponds to $x <_1 y \wedge x <_2 y$.}\label{figure:structureanddatawords}

\end{figure}

Instead of directly verifying conditions imposed by an $\fotwo(<_1, <_2)$-formula, it is convenient to first convert such formulas  into a set of \emph{existential} and \emph{universal constraints}~\cite{SchwentickZ12}. An existential constraint enforces that for each element $x$ of unary type $\sigma$ there is an element $y$ of unary type $\tau$, such that $y$ is in a specified direction from $x$ with respect to $<_1$ and $<_2$, for instance in direction $x <_1 y$ and $x <_2 y$. A universal constraint can forbid patterns, that is, it can state that it is not the case that $x$ and $y$ are elements with unary types $\sigma$ and $\tau$, respectively, and $y$ is in a specified direction of $x$ (see Figure \ref{figure:structureanddatawords}).

Such constraints can be easily verified by a register automaton. To this end, the automaton has registers $r^\sigma_{\max, \leftarrow}$, $r^\sigma_{\min, \leftarrow}$, $r^\sigma_{\max, \rightarrow}$, and $r^\sigma_{\min, \rightarrow}$, for each label $\sigma$, intended to store the maximal and minimal data value of $\sigma$-labeled positions to the left and right of the current position, respectively. The content of these registers can be guessed and verified by the automaton. Then, for determining whether an existential constraint such as the one above is satisfied, the automaton verifies that, for each $\sigma$-labeled position $x$, the register $r^\tau_{\max,\rightarrow}$ stores a value larger than the data value at $x$. Similarly universal constraints can be checked. Technically, the register automaton also has to ensure that all data values are distinct. While this is not possible in general, it can verify a weaker condition that guarantees that if some data word is accepted then so is one with distinct data values.

The above rough sketch can be used to obtain a new proof of decidability of  finite satisfiability problem of $\fotwo$ with two linear orders \cite{SchwentickZ12}. This automata-based approach generalizes well to various other results in this vein, by 
considering various domains $\data$ or various shapes of the input structures. The present paper is an illustration of the power of this approach.

To solve the general satisfiability of $\fotwo(<_1, <_2)$ for linear orders $<_1$ and $<_2$, register automata need to be generalized in two directions. First, to allow infinite domains, we pass from finite data words to infinite objects with data, such as $\omega$-words or infinite trees.
Considering $\omega$-words allows to study the case when $<_1$ is isomorphic to the naturals with their order, which does not cover all infinite orders. To encompass arbitrary countable linear orders, we  move to infinite trees, as any countable order is isomorphic to a subset of the infinite complete binary tree with the left-to-right order on its nodes. More generally, any countable tree order can be encoded in the complete binary tree in a certain sense, so moving to infinite trees allows us to consider arbitrary countable tree orders $<_1$.

For the moment, let us focus on the case when $<_1$ is considered to be isomorphic to the naturals with their order, in which case we consider $\omega$-words with data. Now, existential constraints coming from the $\fotwo$ formula can enforce supremum-like conditions, e.g., they can require that (1) every $\sigma$-labeled element $x$ has a $\tau$-labeled element $y$ with $x <_1 y$ and $x <_2 y$ and vice versa, that is, every $\tau$-labeled element $x$ has a $\sigma$-labeled element $y$ with $x <_1 y$ and $x <_2 y$, and that (2) there is a $\rho$-labeled position that bounds from above all these $\tau$ and $\sigma$-labeled positions with respect to $<_2$. This is why we need 
to consider register automata over $\omega$-words with data, which
can access at each position the infimum and supremum of the values stored in a specified register in all future positions.

Finally, to solve the case of arbitrary countable orders $<_1$, we move from $\omega$-words to infinite trees, as discussed above. Now, the infimum and supremum needs to be taken over all nodes of the tree which are descendants of the current node. 
     This leads us to the study of tree register automata with suprema and infima constraints for infinite, binary trees with data. These are introduced in Section~\ref{section:registerAutomata}.

     For  the variants of the satisfiability problem briefly mentioned above -- whether $<_1$ is  a finite order, or isomorphic to the order of the naturals, or a countable order, or a tree order -- 
     the reduction of the satisfiability problem to the emptiness problem for the corresponding variant of automata always follows essentially the simple idea described above. This is described in Section~\ref{section:FO2toRA}. 
     
     On the other hand, 
      deciding the emptiness problem for the corresponding models of register automata  discussed above becomes more involved as the models are generalized. The overall idea of  deciding emptiness tree register automata with extrema constraints is to (1) reduce to the emptiness problem for tree register automata without extrema constraints and subsequently (2) reduce the emptiness problem for such automata to the emptiness problem of parity automata over infinite trees. This is described in Section~\ref{section:decidability}. Decidability then follows from the fact that the emptiness problem for parity automata is  decidable~\cite{Rabin72}. 
      
This article is the long version of \cite{TorunczykZ20}. It in particular provides full proofs for all results.

\section{Preliminaries}\label{section:preliminaries}
In this section we recapitulate basic notions and fix some of our notations. 

A \emph{tree} $t$ is a prefix-closed subset of $\{0, 1\}^*$, and each element $v \in t$ is called a node. The \emph{ancestor order} of a tree $t$, denoted $<_\text{anc}$, is the strict prefix order on $t$.

We write $\mswarrow v$ and $\msearrow v$ for the left child $v0$ and the right child $v1$ of $v$; and denote the parent of $v$ by~$\muparrow v$. The subtree of $t$ rooted at $v$ is written as $t_v$. Assigning a label from a set $\Sigma$ to every node of $t$ yields a \emph{$\Sigma$-labeled tree}.

A \emph{(strict) partial order $<$} over a domain $\data$ is a transitive, antisymmetric, and antireflexive relation, that is, 
if $a, b, c \in \data$ then $a  <  b$ and $b  <  c$ implies $a < c$; and  $a  <  b$ and $b  <  a$ are never both satisfied at the same time.
We use standard notions of \emph{upper} and \emph{lower bounds}, and of \emph{suprema} and \emph{infima} of subsets of a partially ordered set. 
A partial order is a \emph{linear order} if $a <  b $ or  $b  <  a$ for all $a, b \in \data$ with $a \neq b$. A partial order $(\data, <)$ is a \emph{tree order} if the set $\{y \mid y < x\}$ is totally ordered by $<$ for each $x \in  \data$,
and additionally, any two elements have some lower bound.
For instance, all linear orders are tree orders, and  the ancestor orders of trees are tree orders. 
In general, a tree order may not be isomorphic to the ancestor order of a tree. As examples consider a dense linear order, or an infinitely branching tree,
or a combination of the two, where a dense linear order branches into infinitely many copies at each rational.

In a tree order (or tree) if $u,v$ are two nodes and $u<v$ then we say that $u$ is an \emph{ancestor} of $v$, that $u$ is \emph{smaller} than $v$, and that $v$ is \emph{larger} than $u$.

\section{Tree Register Automata with Suprema and Infima Constraints}\label{section:registerAutomata}
\newcommand{\val}{\textsc{val}}
\renewcommand{\root}{\textsc{root}}

Register automata are finite state automata equipped with a set of registers that can store values from an infinite data domain. Here we introduce a variant of register automata for infinite trees and ordered domains. In the next section we will prove that their emptiness problem is decidable.

Our register automaton model is equipped with a mechanism for accessing infima and suprema of subtrees. Therefore it uses values from the domain $\Rbar \df \R \cup \set{-\infty, \infty}$ linearly ordered in the natural way. The only feature of $\Rbar$  which will matter is that it is a dense linear order and contains all infima and suprema. Instead of $\Rbar$ we could equally well consider the real interval $[0,1]$. We therefore fix the \emph{ordered domain} $\data\df \langle \Rbar, <\rangle$ which is a complete dense linear order with endpoints (treated as constants).

We consider a model of nondeterministic tree automata (with non-deterministic guessing of data values) which process infinite binary trees whose nodes are labeled by a label from a finite alphabet $\Sigma$ and tuples of data values from $\data$, representing input values. Such an automaton has a finite set of registers storing values. It nondeterministically 
assigns to each node of the input tree a state from a finite state space 
and a valuation of the registers in the domain~$\data$
(in particular, it allows for ``guessing'' of data values).
At a node $v$, the transition relation has access to the automaton state as well as the $\Sigma$-label of $v$ and its children $\mswarrow v$,~$\msearrow v$, and it is capable of comparing the input numbers and the numbers stored in registers at those nodes using the linear order. Furthermore it can compare any of those register values to  the infimum or supremum of data values in a given register at nodes in the subtree rooted at $v$ reachable from $v$ by a path whose labels satisfy a given regular property, e.g. of the form \emph{supremum of values of register $r$ 
at all descendants of $v$ reachable by a path labeled by $a^*ba^*$}. 
The transition relation is described by a propositional formula, whose atomic formulas correspond to label and state tests, as well as register comparisons for the current node and its children,
and the suprema and infima of register values in the current subtree.

More formally, a \emph{tree register automaton with suprema and infima constraints} (short: \emph{TRASI}) consists of the following components:
\begin{itemize}
  \item a finite \emph{input alphabet} $\Sigma$;  
	\item a finite set of \emph{states} $Q$;
	\item a finite set of \emph{root states} $F\subset Q$;
	 \item a finite set of \emph{input registers}~$I$;
	\item a finite set of \emph{registers} $R$ containing the input registers~$I$;
	\item a function $\calL$ mapping each register $r \in R$ to an associated regular language over $\Sigma$ (specified by a nondeterministic finite state automaton);
	\item a nondeterministic \emph{transition relation} $\delta$ which is given by a propositional formula with atomic formulas of the form 
          \begin{itemize}
            \item $\sigma$, $\mswarrow \sigma, \msearrow \sigma$ for $\sigma \in \Sigma$ and $q$, $\mswarrow q, \msearrow q$ for $q \in Q$, used for testing labels and states of the current node and its children nodes;
            \item $s < t$ or $s=t$, used for comparing register values or suprema/infima:            
             $s$ and $t$ range over the registers of the current node, the registers of its children nodes, or are domain constants
            or suprema respective infima terms of the form $\sup r$ or $\inf r$, for $r\in R$  (that is, $s, t \in \setof{r, \mswarrow r, \msearrow r, \sup r, \inf r}{  r \in R}$).            
          \end{itemize}

\item a regular acceptance condition, given by a \emph{parity function} $\func \Omega Q \N$ assigning a \emph{rank} to each state.
\end{itemize}
If a TRASI does not use  suprema and infima terms ($\sup r$ or $\inf r$ for $r\in R$) in its transition relation then it is called a \emph{tree register automaton}.

We next define inputs and runs for a tree register automaton $\calA$ with suprema and infima constraints.
An \emph{input tree} for $\calA$ is a complete infinite binary tree, whose vertices are labelled by elements from~$\Sigma \times \data^I$. The labelling by $\data^I$ will form the register assignment for the input registers. 

For a regular language $L\subset \Sigma^*$ an input tree $t$ and two of its nodes $v$ and $w$, we say that  $w$ is an \emph{$L$-descendant} of  $v$ if $w$ is a descendant of $v$ (possibly $v$ itself) and the labels along the path from $v$ to $w$ in $t$ form a word which belongs to $L$. The terms $\sup r$ and $\inf r$, evaluated at a node $v$, will denote the supremum/infimum of all values of register $r$ at nodes $w$ which are $L$-descendants of $v$, where $L$ is the language associated to register $r$.
For technical reasons, we provide a piecemeal definition of a run of a TRASI, in which the values of the registers $\sup r$ and $\inf r$ are stored in auxiliary registers defined below, and then we require that the values of those registers are as expected.

Let $R_{\sup}$ and $R_{\inf}$ be two copies of the set $R$,
where $R_{\sup}=\set{\rsup\mid r\in R}$ and $R_{\inf}=\set{\rinf\mid r\in R}$.
Denote $\cal R=R\cup R_{\sup}\cup R_{\inf}$. Abusing language, elements of $R_{\sup}$ will be called \emph{suprema registers}, and elements of $R_{\inf}$ will be called \emph{infima registers}.

A \emph{pre-run} $\rho$ over an input tree $t$ annotates each node of $t$ by an element from $Q \times \data^{\cal R}$, so that registers from $I  \subseteq \cal R$ get a value according to the input tree. 
The \emph{state} at a node $v$ in a pre-run $\rho$ is the $Q$-component of $\rho(v)$.
The labelling by $\data^{\cal R}$ will form the register assignment for the registers in $\cal R$. We write $r_\rho(v)$ for the value of register $r$ in pre-run $\rho$, or $r(v)$ if $\rho$ is clear from the context.

A pre-run $\rho$ is 
 \emph{locally consistent} if  each node $v$ with left and right children $\mswarrow v$ and $\msearrow v$ satisfies $\delta$ with respect to $\rho$. Here, the satisfaction of the atomic formulas is defined as follows:
\begin{itemize}
  \item $\sigma$, $\mswarrow \sigma$, and $\msearrow \sigma$ are satisfied in 
  $v$ if $v, \mswarrow v$, and $\msearrow v$ are labelled by $\sigma \in \Sigma$, respectively; similarly for $q$, $\mswarrow q$ and $\msearrow q$, for $q \in Q$; 

  \item given a node $v$, a term of the form $r \in R$ evaluates to the value $r(v)$. A term of the form $\mswarrow r$, for $r\in R$, evaluates to the value $r(\mswarrow v)$, and likewise for $\msearrow r$. Finally, a term of the form $\sup r$ evaluates to $\rsup(v)$, and a term of the form $\inf r$ evaluates to $\rinf(v)$.

  \item if $s,t$ are two terms as above, then the atomic formula $s < t$ is satisfied in the node $v$ if the value of the term $s$ is smaller than the value of the term $t$.
  The definition for $s=t$ is analogous.
\end{itemize}
Satisfaction for boolean operations is defined as usual.

A  \emph{run} is a locally consistent pre-run which additionally satisfies the following consistency requirement:
for every node~$v$, $\rsup(v)$ is the supremum 
of the values $r(w)$, where $w$ ranges over all $L$-descendants of $v$ where $L$ is the language associated to $r$,  
whereas $\rinf(v)$ is the infimum of all such values.
If no $L$-descendants exist, then the supremum is $-\infty$ and infimum is $+\infty$ by definition.
A pre-run $\rho$ is \emph{accepting} if the state at the root belongs to $F$, and 
every branch of $\rho$ satisfies the parity condition, i.e., 
on every (infinite, rooted) branch of~$\rho$, 
	if the states at the nodes along this branch are
	$q_1,q_2,\ldots,$ then 
	$\limsup_{n\rightarrow\infty}(\Omega(q_n))$  is even.

A TRASI \emph{accepts} an input tree $t$ if it has an accepting run over that tree.
 The \emph{emptiness problem} for TRASI is the problem of deciding whether  a given TRASI accepts some input tree. In the next section we prove that this problem is decidable.

Note that
	instead of equipping the automaton model with a parity acceptance condition, we could
	equip it with an MSO acceptance condition, i.e.~an MSO formula $\phi$
	using two binary predicates: $s_1$ and $s_2$, standing for left and right successor (in the binary tree), and for each $q\in Q$ a unary predicate $\lambda_q$, holding at nodes labeled with $q$.
	This would not change the expressive power of the tree automata, since parity tree automata have the same expressive power as MSO over infinite binary trees~\cite{Rabin72,Thomas97}.

\begin{example}
\begin{enumerate}[(a)]
 \item  Consider the language of trees with one data value per node such that the values on each infinite path form a strictly increasing sequence of numbers. A TRASI for this language can work as follows. The data value is stored in an input register $r$. The automaton has two states $q_\top$ and $q_\bot$, where $q_\top$ represents the situation where the sequence of the numbers from the root to the current node is increasing; $q_\top$ is the only root state. A path is accepted if $q_\top$ occurs infinitely often, i.e. the parity function $\Omega$ is defined as $\Omega(q_\top)=2$ and $\Omega(q_\bot)=1$.

  The transition relation is defined by 
    \[\big(\mswarrow q_\top \leftrightarrow (q_\top \land r < \mswarrow r )\big) \land \big(\msearrow q_\top \leftrightarrow (q_\top \land r < \msearrow r)\big)\]
  where the first part checks that the state is propagated properly to the left child, and the second part does the same for the right child.

  \item Consider the language of trees over $\Sigma = \{a, b, c\}$ with one data value per node where the data value of each $a$-labeled node equals the supremum of the $b$-labeled nodes below. A TRASI for this language can work as follows. The data value is stored in an input register $r$ with associated language $\Sigma^*b$. The automaton uses two states  $q_\top$ and $q_\bot$, where $q_\top$ represents the situation where all $b$-nodes seen on a path so far adhere to the condition. A path is accepted if $q_\top$ occurs infinitely often.
  
	The transition relation is defined by \sz{check}
    \[\big(\mswarrow q_\top \leftrightarrow (q_\top \land (\neg a\vee (r = \sup r)))\big) \land \big(\msearrow q_\top \leftrightarrow (q_\top \land (\neg a \vee(r = \sup r)))\big).\]
  \end{enumerate}
\end{example}

\section{Deciding Emptiness}\label{section:decidability}        
In this section we prove the following:

\begin{theorem}\label{thm:decidable}
  The emptiness problem for tree register automata with suprema and infima constraints is decidable. 
\end{theorem}

Given a TRASI $\aut A$, we construct (in Sections \ref{section:trasi-to-traA} and \ref{section:trasi-to-traB}) a tree register automaton $\whA$ (without extrema constraints)  such that $\aut A$ has an accepting run if and only if $\whA$ has an accepting run. As emptiness of register automata over infinite data trees is decidable (see Section \ref{section:tra-decidability}), this will yield decidability of the emptiness problem for TRASI.

\subsection{Warm-Up: From TRASI with Trivial Path Languages to Tree Register Automata}\label{section:trasi-to-traA}
For the sake of readability, we first prove how $\whA$ can be constructed when all languages associated with registers in the TRASI $\aut A$ are equal to $\Sigma^*$, that is, if suprema and infima of data values are taken over all nodes of a subtree. The proof for general regular languages is deferred to the following section. In this section, we henceforth assume that all languages associated with registers in a TRASI are equal to $\Sigma^*$.

We would like to construct $\whA$ so that it accepts exactly those pre-runs of $\aut A$ which are actual runs. This requirement, however, is too strong. 
The reason is that applying an arbitrary monotone bijection of $\data$ to an accepting run of a register automaton always yields an accepting run, 
whereas the same property fails for accepting runs of TRASI's, since such a bijection might not preserve suprema and infima. 

For this reason $\whA$ checks a weaker condition. Fix an input tree $t$, a register $r$ and a node $v$; recall that $t_v$ is the set of descendants of $v$ in $t$.  If the supremum of the values of $r$ in $t_v$ is equal to $c$, then either the supremum is attained at some node $w$ of $t_v$, or there must be an infinite path starting from $v$ such that the supremum of $r$ has value $c$ for all nodes along this path.  Then one can find a sequence of nodes $v_1, v_2, \ldots$ on the path that have nodes $w_1, w_2, \ldots$ below them (but not necessarily on the path) whose $r$-values tend to~$c$~(cf.~Fig.~\ref{fig:inf-path}). \begin{figure}
  \includegraphics[scale=0.8,page=1]{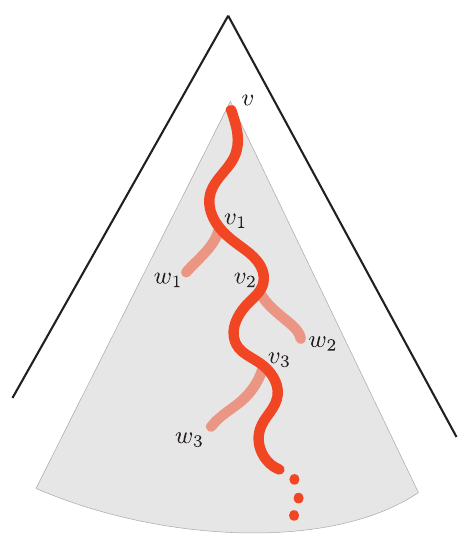}
  \caption{An infinite path and finite paths departing from it,
  forming a witness of a supremum. The path is constructed so that 
   $\sup\set{r(w)\mid w\in t_{v_n}}$ is the same for all $n\ge 0$, 
  and equal to $\sup\set{r(w_n)\mid n\ge 1}$.} 
  \label{fig:inf-path}
\end{figure} Such nodes $w_1, w_2, \ldots$ are called \emph{witness nodes} in the following. The intention of $\whA$ is to check for the existence of these witness nodes. As the automaton cannot check that their data values actually tend to the supremum $c$, it chooses the witnesses $w_1, w_2, \ldots$ so that at least they eventually exceed any values smaller than $c$ which are stored at any point along the path in some register.
This is formalized below.

First, we impose a  normalization assumption on the TRASI $\aut A$.
Namely, say that a TRASI $\aut A$ is \emph{normalized} if the transition relation does not allow directly comparing values in the left and right child of a given node $v$, i.e., it does not contain atoms of the form 
$\msearrow r<\mswarrow s$ or $\mswarrow r<\msearrow s$ or $\mswarrow r=\msearrow s$, for $r,s\in R$.
It is easy to see that any TRASI can be converted into a normalized TRASI which accepts the same input trees (possibly at the cost of introducing new registers).

Let $\rho$ be a pre-run of $\aut A$. Say that $\rho$ is \emph{extrema-consistent} if for every node $v$ and $r\in R$,
\begin{align*}
  \rsup(v)&=\max(r(v),\rsup(\mswarrow v), \rsup(\msearrow v)),\\
  \rinf(v)&=\min(r(v),\rinf(\mswarrow v), \rinf(\msearrow v)).
\end{align*}
Clearly, every run is extrema-consistent, and
 a tree register automaton 
(without extrema constraints) can  easily verify that a given pre-run is extrema-consistent and locally consistent. Because of this,
\begin{quote}\itshape from now on, we assume 
  that every considered TRASI is normalized, and that all pre-runs are  locally consistent and extrema-consistent. 
\end{quote}

The existence of a pre-run as described above does not yet guarantee the existence of a run, since the extrema might not be approached by actual values of registers. For example, it might be the case that $\rsup(v)=5$ for all nodes $v$, whereas $r(v)=4$ for all nodes~$v$. This case can be easily detected by a register automaton, which checks that whenever $r(v)<\rsup(v)$, then there is some descendant $v'$ of $v$ with $r(v)<r(v')<\rsup(v)$. 
Now suppose that  $\rsup(v)=5$ and $r(v)=4-1/d$ for all nodes at depth $d$, which is still not an actual run. Insofar as the register $r$ is concerned, this pre-run can be modified into a run by simply replacing $\rsup(v)$ by $4$. However, suppose there is another register $s$ such that $s(w)=4.5$ for some $w$ and that the inequalities $r(w)<s(w)<\rsup(w)$ are enforced by the acceptance condition. Then we cannot easily fix this pre-run to obtain a run, since $s(w)$ separates the values $r(v)$ from the value $\rsup(v)$. To make sure that such a situation does not occur, we introduce witness families which ascertain that the values $r(v)$ exceed all values $s(w)$ which are encountered along a branch.

Intuitively, a witness family partitions the nodes of a tree into finite paths such that each such path provides a witness for an extremum of some register. To this end, paths of such a partition are labeled by $\rw+$ or $\rw-$, indicating whether they are the supremum or infimum witness for a register $r$.  This is formalized as follows.

We introduce two symbols $\rw+$ and~$\rw-$ for each register $r\in R$. Let $R^{\pm}=\setof{\rw+,\rw-}{r\in R}$ be the collection of all such symbols, and fix an arbitrary total order on $R^{\pm}$. Further fix a partition $\cal P$ of the set of nodes of the tree underlying a pre-run $\rho$ into finite paths (cf. Fig.~\ref{fig:partition}). Say that a path $\pi$ in $\cal P$ is the \emph{father} of a path $\pi'$ if the father of the smallest node (wrt. the ancestor order) in $\pi'$ belongs to $\pi$. 
\begin{figure}
  \includegraphics[scale=1,page=2]{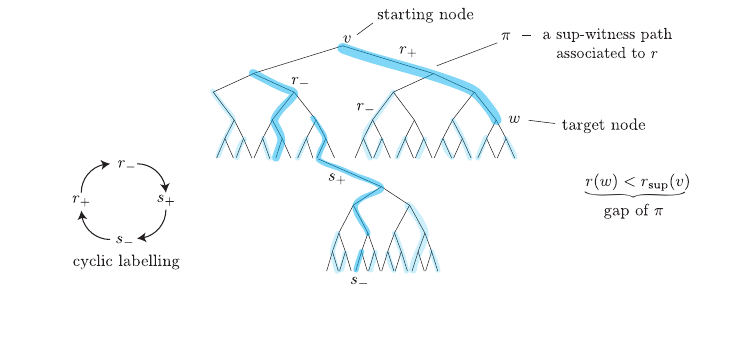}
  \caption{A partition of the tree into witness paths, a labelling of the paths with $R^{\pm}$ and the gap of a witness path.} 
  \label{fig:partition}
\end{figure}

Let $\lambda\from \cal P\to R^{\pm}$ be a labeling of the paths in $\cal P$. A path $\pi\in \cal P$ is called a \emph{$\sup$-witness path}  
associated to $r$ if $\lambda(\pi)=\rw+$ and it is called a \emph{$\inf$-witness path} associated to $r$ if $\lambda(\pi)=\rw-$.
For a witness path $\pi$ associated to $r$,
its \emph{target value} is $r(w)$, where $w$ is the largest node in $\pi$ wrt. the ancestor order. A witness path \emph{starts} at its smallest node wrt. the ancestor order. 
If $\pi$ is a witness path associated to $r$ which starts at $v$, then the \emph{associated extremum} is $\rsup(v)$ if $\pi$ is a $\sup$-witness path, and $\rinf(v)$ if $\pi$ is an $\inf$-witness path. Note that if $\lambda(\pi)=\rw+$ then the target value of $\pi$ is at most equal to the associated extremum of $\pi$, since $\rho$ is extrema-consistent (if $\lambda(\pi)=\rw-$, then the inequality is reversed). 

The \emph{gap} of a witness path with target value $c$ and associated extremum $d$ is the interval with endpoints $c$ (inclusively) and $d$ (exclusively).
It is empty precisely when the target value attains the extremum.
Intuitively, the length of the gap is a measure of the quality 
of the witness path, shorter gaps being better witnesses.

Our goal is to state properties of a partition $\calP$ of a pre-run and a labeling $\lambda\from \cal P\to R^{\pm}$ that guarantee the existence of a run (i.e., a pre-run that is consistent with respect to the extrema). Such a combination of a partition and a labeling will be a certificate for the existence of a run. For the case where all associated languages are $\Sigma^*$, periodically witnessing extrema for each register with close enough target values suffices as a certificate. 

A mapping
$\lambda\from \cal P\to R^\pm$ to labels from an ordered set $R^\pm$ is a \emph{cyclic labeling} if:
\begin{itemize}
  \item the label $\lambda(\pi)$ of the path $\pi$ containing the root 
  is the smallest label in $R^\pm$, and
  \item if a path $\pi$ is the father of $\pi'$ then the label $\lambda(\pi')$ is the successor of the label $\lambda(\pi)$, in cyclic order, according to the fixed order on $R^\pm$.
\end{itemize}
Clearly, there is exactly one cyclic labeling of $\cal P$ with labels from $R^\pm$.

Henceforth (in this section), when considering a partition $\cal P$ of a partial run $\rho$, we assume that each path in $\cal P$ is labeled according to the cyclic labeling with labels from $R^{\pm}$ assuming an arbitrary order on $R^{\pm}$. 

We now define when a partition $\cal P$ is a certificate. Intuitively, we require that the target values of the witness paths are sufficiently close to the associated extrema, i.e., each gap is sufficiently small. Formally, the partition $\cal P$  is a \emph{certificate} if each witness path $\pi$
satisfies the following \emph{gap condition}:
\begin{quote}\itshape
the gap of $\pi$ does not contain any value $c\in \data$ such that
 $c=s(v)$ or $c=s(\uparrow v)$ for some $s\in\cal R$, where $v$ denotes the starting node of $\pi$.
\end{quote}
Below we prove the following three properties of certificates:
\begin{enumerate}  
  \item every run has a certificate,
  \item if some accepting pre-run of $\aut A$ has a certificate then $\aut A$ has some accepting run,
 \item a tree register automaton can verify if a given pre-run has a certificate.
\end{enumerate}
These properties together easily yield Theorem~\ref{thm:decidable} for the case when all associated languages are trivial.

\medskip

To show that every run of $\calA$ has a certificate, we greedily add paths 
to the family $\cal P$ without violating the gap condition, as shown below.

\begin{lemma}\label{lem:cor}Every run of $\calA$ has a certificate.
\end{lemma}
\begin{proof}
  Fix a run $\rho$. 
  We construct the family $\cal P$ in stages. Initially, $\cal P$ is empty, and each stage proceeds as follows.
Pick a minimal node $v$ such that $v\notin \bigcup \cal P$.
If the node $v$ is the root then let $l\in R^{\pm}$ be the smallest label in $R^{\pm}$. Otherwise, the parent of $v$ already belongs to some path $\pi\in\cal P$ which has an assigned label $l'$, and let $l\in R^{\pm}$ be the successor of $l'$ in cyclic order, according to the fixed order on $R^{\pm}$. Assume $l=\rw+$ for some $r\in R$. The  case when $l=\rw-$ is treated symmetrically. 
 Pick any descendant $w$ of $v$ such that
the interval with endpoint $r(w)$ (inclusively) and $\rsup(v)$ (exclusively)
does not contain any value $c\in \data$ such that $c=s(v)$ or $c=s(\uparrow v)$ for some $s\in\cal R$.
Add the path joining $v$ with $w$ to the family $\cal P$ and associate the label $l$ with it.
Proceed to the next stage.

By construction, in the limit we obtain a certificate.
\end{proof}

The key point of certificates is that
the existence of a certificate implies the existence of a run.

\begin{lemma}\label{lem:tight}For every (locally consistent, extrema-consistent) pre-run $\rho$ 
  which has a certificate  there is a run $\rho'$ whose states agree with the states of $\rho$.
\end{lemma}

  \begin{proof}
  To prove the lemma, fix a pre-run $\rho$ and some certificate~$\cal P$.
  We construct a run $\rho'$ which has the same states as $\rho$.
  The run $\rho'$ is obtained by successively processing all nodes $v$, starting from the root, and shifting the register values in their subtrees, without changing local relationships. Towards this goal, let $v_0, v_1, \ldots$ be an enumeration of the tree nodes in which every  node appears after all its ancestors (in particular, $v_0$ is the root). We construct a sequence  ${\rho}_{v_0}, {\rho}_{v_1}, \ldots$ of pre-runs such that each run has the same states as ${\rho}$, and such that the sequence converges to a run ${\rho'}$. Here, convergence means that for every node $v$, the sequence  ${\rho}_{v_0}(v), {\rho}_{v_1}(v), \ldots$ of labels in $Q\times \calR$ assigned to $v$ 
  is ultimately equal to  ${\rho'}(v)$.
  
  Define ${\rho}_{v_0} \df {\rho}$. We describe, for $v \df v_i$ and $u \df v_{i-1}$, how ${\rho}_{v}$ is constructed, assuming that ${\rho}_u$ has already been constructed. The register values of the subtree $t_{v}$ will be shifted, and the assignment will not change on nodes from~\mbox{$t \setminus t_{v}$}. This guarantees that the sequence ${\rho}_{v_0}, {\rho}_{v_1}, \ldots$ indeed converges to some pre-run ${\rho'}$.

The pre-run ${\rho}_v$ is defined as follows. Outside of the subtree $t_v$, it agrees with the pre-run~${\rho}_u$. Inside the subtree $t_v$, an arbitrary monotone bijection $f\from \data\to\data$ with the following properties is applied to all register values of nodes in $t_v$ in the pre-run~${\rho}_u$:
    \begin{enumerate}
      \item $f(c)=c$ for all $c\in\data$ such that 
       $c=s(v)$ or $c=s(\uparrow v)$, for some $s\in\cal R$, and        
      \item if $c\in\data$ is the target value of a witness path $\pi$ starting at $v$ with associated extremum $d$, then
     $|f(c)-f(d)|<\frac 1 n$, where $n$ is the depth of the node $v$ 
(note that $c=f(c)$ by the first item).
    \end{enumerate}
    Intuitively, the bijection $f$ shrinks the gap of $\pi$ so that its length is smaller than~$\frac 1 n$. As the gap of $\pi$ did not contain any values of the form $s(v)$ or $s(\uparrow v)$ by definition of a certificate,
    a monotone bijection $f$ with the above properties exists. 
This ends the description of $\rho_v$.

Note that after applying the bijection $f$ to the values inside $t_v$,
    the gap of $\pi$ still does not contain any values of the form 
    $s(v)$ or $s(\uparrow v)$. Hence, $\cal P$ is a certificate for $\rho_v$.
Therefore, the above construction preserves the following invariant:
\begin{itemize}
  \item $\rho_v$ is a (locally consistent, extrema-consistent) pre-run with the same states as $\rho$,
  \item $\rho_v$ has certificate $\cal P$.
\end{itemize}
It follows that the sequence of pre-runs ${\rho}_{v_0}, {\rho}_{v_1}, \ldots$ converges to a pre-run $\rho'$ with the above properties.
It remains to show that $\rho'$ is a run, i.e., 
that for every node $v$ and register~$r$,
\begin{align}
  \rsup(v)&=\sup\setof{r(v')}{v'\textit{ is a descendant of }v}\label{ineq:sup}\\
  \rinf(v)&=\inf\setof{r(v')}{v'\textit{ is a descendant of }v}\label{ineq:inf}.
\end{align}
We prove~\eqref{ineq:sup} while~\eqref{ineq:inf} is proved analogously. 
The inequality $\ge$ in~\eqref{ineq:sup} follows from  extrema-consistency of $\rho'$. 

Towards proving $\leq$, pick an inclusion-maximal branch $\pi$ starting at $v$, such that 
$\rsup(v)=\rsup(v')$ for all $v'\in \pi$.
If the branch $\pi$ is finite then $\rsup(v)=r(w)$ where $w$ is the largest node in $\pi$ (wrt. the ancestor order) and $\rsup(\swarrow w),\rsup(\searrow w)$ are both smaller than $r(w)$. By extrema-consistency of $\rho'$ it follows that the right-hand side in~\eqref{ineq:sup} is equal to $r(w)$, which proves~\eqref{ineq:sup} as $\rsup(v)=r(w)$. 

Suppose now that $\pi$ is infinite. 
To prove~\eqref{ineq:sup}, it suffices to exhibit a sequence $w_1,w_2,\ldots$ of descendants of $v$ such that 
\begin{align}\label{eq:lim}
\rsup(v)=\lim_{n\rightarrow+\infty} r(w_n).
\end{align}

Construct $\sup$-witness paths $\pi_0,\pi_1,\ldots$ associated to $r$, as follows. Assuming we have constructed $\pi_0,\pi_1,\ldots,\pi_{n-1}$, the path $\pi_n$ is any $\sup$-witness path associated to $r$ such that its  starting node $u_n$ is in $\pi$, is larger than $u_{n-1}$, and does not belong to $\pi_0\cup\ldots\cup \pi_{n-1}$. 
It is easy to see that such a node $u_n$ exists, as $\cal P$ is a partition into finite paths, and the labeling is cyclic.

Then the nodes $v=u_0,u_1,\ldots$ all lie on the path $\pi$. In particular, $\rsup(v)=\rsup(u_n)$ and $u_n$ has depth at least $n$, for all $n$. 
Fix $n$ and let $w_n$ be the largest node in $\pi_n$ wrt. the ancestor order.
By  the second item in the definition of the bijection $f$ obtained when defining the pre-run $\rho_{u_n}$, the following holds for that pre-run:
\begin{align}\label{eq:lim'}
|r(w_n)-\rsup(v)|=|r(w_n)-\rsup(u_n)|<\frac 1 n.
\end{align}
We observe that the same inequality holds for the limit pre-run~$\rho'$.
This follows from the first item in the definition of the bijection $f$ obtained when defining the pre-runs $\rho_{w}$ for all $w$ which are descendants of $u_n$.

Since the inequality~\eqref{eq:lim'} holds for the pre-run $\rho'$ and for all $n$, this proves~\eqref{eq:lim}, yielding~\eqref{ineq:sup}. Hence $\rho'$ is  a run.
   \end{proof}

   It remains to show that the existence of certificates can be decided by tree register automata. This is proved by encoding the data  
   of a certificate for a pre-run $\rho$ using a finite labeling (for marking starting points and labels of witness paths) and one register (storing at a node $v$ the target value of the witness path containing $v$) and verifying that they form a valid certificate:
   a parity condition for checking that the paths are finite, and 
    inequality constraints on the target values of the witness paths to verify that they satisfy the gap condition of a certificate.
   
   \begin{lemma}\label{lem:dec-para}
     There is a register automaton $\aut B$ which accepts a pre-run $\rho$ of  $\aut A$ if and only if it has a certificate. Moreover, $\aut B$ can be  constructed in polynomial time, given $\cal R=R\cup R_{\sup}\cup R_{\inf}$.
   \end{lemma}
   \begin{proof}[Proof sketch]
     The automaton $\aut B$ guesses a certificate for $\rho$, as follows.
   A partition $\cal P$ of all vertices into connected sets can be represented as the set $X=\set{\min P\mid P\in\cal P}$ of least elements of the sets in $\cal P$. Then $\cal P$ is a partition into finite paths if and only if every branch contains infinitely many elements of $X$ and every node has at least one child in $X$.
   
   The automaton $\aut B$ guesses the partition $\cal P$,  represented by the set $X$ as described above, and verifies that it is a partition into finite paths.
   Additionally, each node is nondeterministically labeled by a label in $R^{\pm}$ and the automaton verifies that nodes in the same part of $\cal P$ have equal labels, and that this yields a cyclic labeling of $\cal P$. 
   \newcommand{\wit}{t}
      The automaton has a single register $\wit$.
      At each node, it nondeterministically selects a value 
      for this register and 
      verifies that the value $\wit(v)$ is the same for all nodes $v$ in the same part of $\cal P$, and that $\wit(v)=r(v)$ if $v$ is the largest node in its part (wrt. the ancestor order)
      and is labeled 
      $\rw+$ or~$\rw-$.
      Hence, $\wit(v)$  represents the target value of the witness path associated to $r$ and containing $v$.
      
      It remains to verify that the guessed partitions form a certificate for $\rho$. To this end, the automaton 
      verifies, for each node $v$ and register $r$, that 
   the interval with endpoint $\wit(v)$ (inclusively) and $\rsup(v)$ (exclusively)
   does not contain any value $c\in \data$ such that
    $c=s(v)$ or $c=s(\uparrow v)$ for some $s\in\cal R$.
   
   If all the above is confirmed, the automaton accepts. By construction, $\aut B$ accepts $\rho$ if and only if $\rho$ has a certificate.
   \end{proof}

The special case of Theorem~\ref{thm:decidable} where all associated languages are $\Sigma^*$ now follows easily via a reduction from the emptiness problem for TRASI to the emptiness problem for tree register automata. That the latter problem is decidable is shown in Section \ref{section:tra-decidability}.

Fix a TRASI $\aut A$. We construct a tree register automaton $\whA$ such that $\aut A$ has an accepting run if and only if $\whA$ has an accepting run. First, construct the tree register automaton $\aut B$ as given by  Lemma~\ref{lem:dec-para}. Then, construct an automaton $\whA$ which is the composition of $\aut A$ and $\aut B$:
given a labeled input tree $t$, the automaton $\whA$  nondeterministically selects  a pre-run $\rho$ of $\aut A$, tests that it is locally consistent, extremum-consistent, satisfies the parity condition, and uses $\aut B$ to verify that $\rho$ has a certificate. If all those tests are passed, then $\whA$ accepts. Note that $\whA$ is a tree register automaton (without extrema constraints).

Then $\whA$ has an accepting run if and only if $\aut A$ has an accepting run. 
In one direction, suppose $\aut A$ has an accepting run $\rho$ on some input tree $t$. Then $\whA$ also accepts $t$, as witnessed by the pre-run $\rho$, which is accepted by $\aut B$ by Lemma~\ref{lem:cor}.
Conversely, suppose that $\whA$ has an accepting run.
This means that there is a pre-run $\rho$ of $\aut A$ which is extrema-consistent, satisfies the parity condition, and has a certificate. By Lemma~\ref{lem:tight}, there is a run $\rho'$ of $\aut A$ whose states agree with the states of $\rho$. In particular, $\rho'$ satisfies the parity condition, so is an accepting run of $\aut A$. 

This completes the reduction, and proves Theorem~\ref{thm:decidable} in the special case where the languages associated to the registers are trivial.

\subsection{From TRASI to Tree Register Automata}\label{section:trasi-to-traB}
In this section, we lift the construction from the previous section to TRASIs where languages associated with registers can be arbitrary regular languages. 

Suppose that $\aut A$ is a TRASI. As before, we intend to construct a tree register automaton $\whA$ that  checks conditions on a pre-run that ensure that there is an actual run of $\aut A$. Again, we assume that the TRASI $\aut A$ is normalized.

We first adapt extrema-consistency for the case of arbitrary associated regular languages. Consider a register $r$ with associated language $L$. If the supremum of $r$ at a $\sigma$-node $v$ is $c$ then either (1) $\varepsilon \in L$ and $r$ has value $c$ at node $v$, or (2) $c$ is the supremum of the $r$-values of $L'$-descendants of the left child of $v$, where $L' \df \{x \mid \sigma x \in L\}$, or (3) likewise for the right child of $v$. It will therefore be convenient to assume that there are registers with such associated languages $L'$.

For a regular language $L$ there are only finitely many languages $u^{-1}L \df \{x \mid ux \in L\}$ for $u \in \Sigma^*$. We say that $\aut A$ is \emph{language-closed} if whenever a register $r$ has $L$ as associated language, then there is also a register with associated language $u^{-1}L$ for all $u \in \Sigma^*$. If $r$ is a register with associated language $L$, then we denote by $ur$ the register with associated language $u^{-1}L$.

From now on, we assume without loss of generality that $\aut A$ is language-closed.

Let $\rho$ be a pre-run of $\aut A$. Say that $\rho$ is \emph{regular extrema-consistent} if for every node $v$ labeled with $\sigma \in \Sigma$ and $r\in R$ with associated language $L$,
\begin{align*}
  \rsup(v)&=\max(d_{\text{sup}},\sigma\rsup(\mswarrow v), \sigma\rsup(\msearrow v))\\
  \rinf(v)&=\min(d_{\text{inf}},\sigma\rinf(\mswarrow v), \sigma\rinf(\msearrow v))
\end{align*}
where $d_{\text{sup}} \df  r(v)$ if $\varepsilon \in L$ and $d_{\text{sup}} = - \infty$ otherwise, and similarly for $d_{\text{inf}}$.

Clearly, every run of $\aut A$ is regular extrema-consistent, and a tree register automaton 
(without extrema constraints) can  easily verify that a given pre-run is extrema-consistent and locally consistent. Because of this,
\begin{quote}\itshape from now on, we assume 
  that every considered TRASI is normalized, and that all pre-runs are  locally consistent and regular extrema-consistent. 
\end{quote}

Next we adapt certificates. We introduce witness families that ensure that for each node $v$ and each register $r$ with associated language $L$, the supremum of $r$ at $v$ is approached by values of $r$ of $L$-descendants of $v$. Essentially, again, this will be certified by a path starting from $v$ on which there is a sequence of nodes $v_1, v_2, \ldots$ with nodes $w_1, w_2, \ldots$ below them (but not necessarily on the path) such that (1) there is an $L$-path from $v$ to all $w_i$ and (2) the $r$-values of the $w_i$ tend to $c$. The paths from $v_i$ to $w_i$ will be paths in the witness family. The existence of such witness families will be verifiable by tree register automata (without suprema constraints). 

A complication arises from using arbitrary regular languages. Suppose that~\mbox{$\calB = (Q, \delta, s, F)$} is a deterministic finite state automaton for $L$.  If a path from $v$ to $v_i$ brings $\calB$ from its initial state to state $q$, then the path from $v_i$ to $w_i$ has to be labeled by a word in the language $L(B_q)$ of $B_q \df (Q, \delta, q, F)$ to ensure that the path from $v$ to $w_i$ is in $L$. Furthermore, the extrema of all nodes $v$ along a path $\pi$ have to be witnessed as, intuitively, the automaton $\calB$ is ``started'' at each node $v \in \pi$.

As in the previous section, our goal is to state properties of a partition $\calP$ of a pre-run of our TRASI $\calA$ and a labeling $\lambda\from \cal P\to R^{\pm}$ that guarantee the existence of a run of $\calA$ on the same input tree (i.e., a pre-run that is consistent with respect to the suprema). For the case where all associated languages may be arbitrary regular languages, the labeling is slightly more involved and depends on the associated languages. 

For ensuring witnesses for a register $r$ with associated language $L$ for all nodes, we consider trees that encode all runs of an automaton for $L$ along paths of an input tree. 
Formally, let $\calB = (Q, \delta, s, F)$ be a deterministic finite state automaton and $<$ an arbitrary order on $Q$. 
Let $|Q|=m$; the 
elements of $\set{1,\ldots,m}$ will be called \emph{tapes}.
The \emph{run tree} $\widehat{t}$ of a tree $t$, defined below, associates to each tape and node $v$ a state in $Q$. The idea is that tapes encode runs of $\calB$ along the path. However, whenever two tapes $i < j$ enter the same state, only the smaller tape $i$ continues this run; the other tape starts a new run from an ``unused'' state. 
See Fig.~\ref{fig:tapes} for the construction of the run tree in the case when
$t$ is a word.
This construction originates from \cite[Section~2]{DBLP:conf/dlt/Bojanczyk09}, where it is applied to words.
The construction for trees is performed analogously, by carrying out the same construction on each branch.

\begin{figure}
  \includegraphics[scale=0.8,page=3]{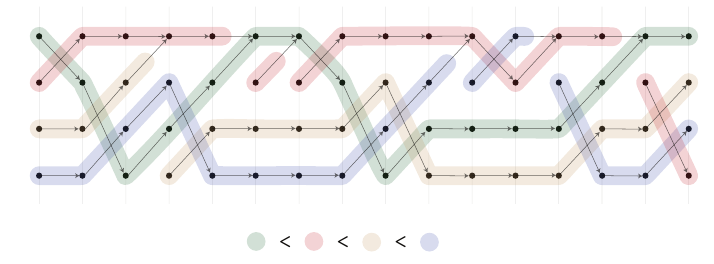}
  \caption{Construction of tapes along an input word.
  The automaton has $4$ states and is deterministic. 
  The dots are pairs consisting of a position in the word and a state of the automaton.
  The arrows denote the transitions of the automaton along the input letters. The colors correspond to the tapes.
  Every color occurs exactly once at each position, and when following an edge, the color can either remain unchanged or decrease.
  } 
  \label{fig:tapes}
\end{figure}

Formally the \emph{run tree}  $\widehat{t}$  of $\calB$ for a tree $t$ associates a tuple $\bar q\in Q^{m}$ to each node $v$ of $t$ as follows. In such a tuple $\bar q$, no two coordinates carry the same state; in particular all states from $Q$ appear in each such tuple.  
The root is annotated by $(q_1, \ldots, q_m)$ with $q_i < q_{i+1}$. If a node is labeled by $\sigma \in \Sigma$ and annotated by $\bar q = (q_1, \ldots, q_m) \in Q^m$ then its children are annotated by the permutation $\bar p = (p_1, \ldots, p_m)$ obtained according to the above intuition as follows. Let $\bar p' = (p'_1, \ldots, p'_m)$ be such that $p'_i = \delta(q_i, \sigma)$. Then $\bar p$ is constructed from  $\bar p'$ by only keeping the smallest copies of identical states and replacing the other copies by states not used in $\bar p'$ in increasing order. Formally the components of $\bar p$ are defined inductively for $i = 1, \ldots, m$:
\begin{itemize}
 \item if $p'_i$ is not equal to $p_1, \ldots, p_{i-1}$ then $p_i \df p'_i$,
 \item else $p_i$ is defined to be the smallest state not occuring in $\bar p'$ and $p_1, \ldots, p_{i-1}$.
\end{itemize}

The \emph{run tree of a register} $r$ is the run tree of a deterministic finite state automaton for the language associated with $r$.

We now state properties of a partition $\calP$ of a pre-run of our TRASI $\calA$ and a labeling $\lambda\from \cal P\to R^{\pm}$ that guarantee the existence of a run on the same input tree.

Intuitively, a regular labeling $\lambda\from \cal P\to R^{\pm}$ ensures that all ``parallel'' runs of an automaton for a language associated to some register are witnessed periodically.  Formally, let 
\[K^{\pm} \df \{(r^+, i), (r^-, i) \mid  r\in R, i\in\set{1,\ldots,m}\}\]
and fix an arbitrary total order on $K^{\pm}$.

A labeling $\lambda\from \cal P\to R^{\pm}$ is \emph{regular} if it is obtained from a cyclic labeling $\kappa\from \calP \to K^{\pm}$ as follows. Consider a path $\pi$ with $\kappa(\pi) = (r_+, i)$ and smallest node $v$. Suppose that the run tree of $r$ annotates $v$ with $(q_1, \ldots, q_m)$, then $\lambda(\pi)$ is $t_+$ where $t$ is the register with associated language $L(\calB_{q_i})$. Analogously for $\kappa(\pi) = (r_-, i)$. Intuitively, the path $\pi$ is a witness for all ancestors of $v$ for which the automaton for the language associated with $r$ reaches state $q_i$ at node $v$.

Henceforth (in this section), when considering a partition $\cal P$ of a partial run $\rho$, we assume that each path in $\cal P$ is labeled according to the (unique) regular labeling with labels from~$R^{\pm}$.

We now define when a partition $\cal P$ is a regular certificate. The partition $\cal P$  is a \emph{regular certificate} if each witness path $\pi$ satisfies the following \emph{gap condition}:
\begin{quote}\itshape
the gap of $\pi$ does not contain any value $c\in \data$ such that
 $c=s(v)$ or $c=s(\uparrow v)$ for some $s\in\cal R$, where $v$ denotes the starting node of $\pi$.
\end{quote}
As previously, we prove the following three properties of regular certificates:
\begin{enumerate}  
  \item every run has a regular certificate,
  \item if some accepting pre-run of $\aut A$ has a regular certificate then $\aut A$ has some accepting run,
 \item a tree register automaton can verify if a given pre-run has a regular certificate.
\end{enumerate}
These properties together easily yield Theorem~\ref{thm:decidable}. Their proofs are adaptions of the proofs Lemmata \ref{lem:cor},  \ref{lem:tight} and  \ref{lem:dec-para} from the previous section. For the sake of completeness, we repeat the essential proof steps.

\medskip

\begin{lemma}\label{lem:corG}
	Every run of $\calA$ has a regular certificate.
\end{lemma}
\begin{proof}
  Fix a run $\rho$. 
  We construct the family $\cal P$ in stages. Initially, $\cal P$ is empty, and each stage proceeds as follows.
Pick a minimal node $v$ such that $v\notin \bigcup \cal P$.
If the node $v$ is the root then let $k\in K^{\pm}$ be the smallest label in $K^{\pm}$. Otherwise, the parent of $v$ already belongs to some path $\pi\in\cal P$ which has an assigned label $k'$, and let $k\in K^{\pm}$ be the successor of $k'$ in cyclic order, according to the fixed order on $K^{\pm}$. Assume $k=(\rw+, i)$ for some $r\in R$. The  case when $k=(\rw-, i)$ is treated symmetrically. Suppose that the run tree of $r$ annotates $v$ with $(q_1, \ldots, q_m)$, then $l$ is the label $t_+$ where $t$ is the register with associated language $L(\calB_{q_i})$. 
 Pick any descendant $w$ of $v$ such that
the interval with endpoint $t(w)$ (inclusively) and $\rsup(v)$ (exclusively)
does not contain any value $c\in \data$ such that $c=s(v)$ or $c=s(\uparrow v)$ for some $s\in\cal R$.
Add the path joining $v$ with $w$ to the family $\cal P$ and associate the label $l$ with it.
Proceed to the next stage.

By construction, in the limit we obtain a regular certificate.
\end{proof}

\begin{lemma}\label{lem:tightG}For every (locally consistent, regular extrema-consistent) pre-run $\rho$ 
  which has a regular certificate  there is a run $\rho'$ whose states agree with the states of $\rho$.
\end{lemma}

  \begin{proof}
  To prove the lemma, fix a pre-run $\rho$ and some regular certificate~$\cal P$.
  A run $\rho'$ which has the same states as $\rho$ is constructed exactly as in the proof of Lemma \ref{lem:tight}.
  
	The difference with the proof of Lemma \ref{lem:tight} is in showing that $\rho'$ is a run, i.e., that for every node $v$ and register~$r$ with associated language $L$,
\begin{align}
  \rsup(v)&=\sup\setof{r(v')}{v'\textit{ is an $L$-descendant of }v}\label{ineq:supG}\\
  \rinf(v)&=\inf\setof{r(v')}{v'\textit{ is an $L$-descendant of }v}\label{ineq:infG}.
\end{align}
We prove~\eqref{ineq:supG} while~\eqref{ineq:infG} is proved analogously. 
The inequality $\ge$ in~\eqref{ineq:supG} follows from the fact that $\rho'$ is extrema-consistent. 

Denote by $\vartheta(v, v')$ the word between the nodes $v$ and $v'$. Recall that $\calA$ is language-closed and therefore for every register $r$ with associated language $L$ there is a register $\vartheta r$ with associated language $\vartheta^{-1}L$, for all $\vartheta \in \Sigma^*$. 

Towards proving $\leq$, pick an inclusion-maximal branch $\pi$ starting at $v$, such that 
$\rsup(v)=\vartheta(v, v')\rsup(v')$ for all $v'\in \pi$. If the branch $\pi$ is finite then $\rsup(v)=\vartheta(v, w)r(w)$ where $w$ is the largest node in $\pi$ (wrt. the ancestor order) and $\sigma\rsup(\swarrow w),\sigma\rsup(\searrow w)$ are both smaller than $r(w)$ and $\sigma$ is the label of $w$. By regular extrema-consistency of $\rho'$ it follows that the right-hand side in~\eqref{ineq:supG} is equal to $\vartheta(v, w)r(w)$, which proves~\eqref{ineq:supG} as $\vartheta(v, w)\rsup(v)=r(w)$. 

Suppose now that $\pi$ is infinite. 
To prove~\eqref{ineq:supG}, it suffices to exhibit a sequence $w_1,w_2,\ldots$ of descendants of $v$ such that 
\begin{align}\label{eq:limG}
\rsup(v)=\lim_{n\rightarrow+\infty} r(w_n).
\end{align}

Construct $\sup$-witness paths $\pi_0,\pi_1,\ldots$ starting at nodes $v=u_0, u_1, \ldots$ associated to registers $\vartheta(v, u_i)r$, as follows. Assuming we have constructed $\pi_0,\pi_1,\ldots,\pi_{n-1}$, the path $\pi_n$ is any $\sup$-witness path associated to $\vartheta(v, u_n)r$ such that its starting node $u_n$ is in $\pi$, is larger than $u_{n-1}$, and does not belong to $\pi_0\cup\ldots\cup \pi_{n-1}$. 
It is easy to see that such a node $u_n$ exists, by the fact that $\cal P$ is a partition into finite paths, and the fact that their labeling is regular.

Then the nodes $v=u_0,u_1,\ldots$ all lie on the path $\pi$. In particular, $\rsup(v)=\vartheta(v, u_n)\rsup(u_n)$ and $u_n$ has depth at least $n$, for all $n$. 
Fix $n$ and let $w_n$ be the largest node in $\pi_n$ wrt. the ancestor order.
By  the second item in the definition of the bijection $f$ obtained when defining the pre-run $\rho_{u_n}$, the following holds for that pre-run:
\begin{align}\label{eq:lim'G}
|\vartheta(v, w_n)r(w_n)-\rsup(v)|=|\vartheta(v, w_n)r(w_n)-\rsup(u_n)|<\frac 1 n.
\end{align}
We observe that the same inequality holds for the limit pre-run~$\rho'$.
This follows from the first item in the definition of the bijection $f$ obtained when defining the pre-runs $\rho_{w}$ for all $w$ which are descendants of $u_n$.

Since the inequality~\eqref{eq:lim'G} holds for the pre-run $\rho'$ and for all $n$, this proves~\eqref{eq:limG}, yielding~\eqref{ineq:supG}. Hence $\rho'$ is  a run.
   \end{proof}

\begin{lemma}\label{lem:dec-paraG}
     There is a register automaton $\aut B$ which accepts a pre-run $\rho$ of  $\aut A$ if and only if it has a regular certificate. \end{lemma}
   \begin{proof}[Proof sketch]
     The automaton $\aut B$ guesses a certificate for $\rho$, as follows.
   A partition $\cal P$ of all vertices into connected sets can be represented as the set $X=\set{\min P\mid P\in\cal P}$ of least elements of the sets in $\cal P$. Then $\cal P$ is a partition into finite paths if and only if every branch contains infinitely many elements of $X$ and every node has at least one child in $X$.
   
   The automaton $\aut B$ guesses the partition $\cal P$,  represented by the set $X$ as described above, and verifies that it is a partition into finite paths.
   Additionally, each node is nondeterministically labeled by a label from $K^{\pm}$ and from $R^{\pm}$ and the automaton verifies that nodes in the same part of $\cal P$ have equal labels, and that this yields a regular labeling of $\cal P$ by simulating the finite state automata for the languageas associated to registers. 
\end{proof}
As previously,  the properties (1)-(3) listed above and which have now been proved, yield a reduction from the emptiness problem for TRASI to the emptiness problem for tree register automata. It remains to show that the latter problem is decidable.

\subsection{Emptiness for Tree Register Automata is Decidable}\label{section:tra-decidability}

The last step in the proof of Theorem \ref{thm:decidable} is an emptiness test for tree register automata. The proof is along the same lines as the proof for register automata on finite words. The result is folklore, but we provide the proof for the sake of completeness.

\begin{lemma}\label{lemma:emptiness_TRA}
  Emptiness of tree register automata is decidable. 
\end{lemma}

We will show that one can effectively construct, from each tree register automaton $\calA$, a parity tree automaton $\calA'$ such that $\calA$ is non-empty if and only if $\calA'$ is non-empty. One obvious problem is that $\calA$ is over an infinite alphabet, while $\calA'$ needs to work over a finite alphabet. This problem can be tackled by a standard technique: the automaton $\calA'$ abstracts the alphabet and the data values of an input tree for $\calA$ into \emph{types} that capture the essential information. 

For preciseness, let $\calA$ be a tree register automaton with input alphabet $\Sigma$ and input registers~$I$. As we are interested in emptiness, we may assume that $I$ contains all registers of~$\calA$.
Then each node of an input tree $t$ for $\calA$ is labelled by an element from $\Sigma \times \data^I$. 

Consider a node $v$ of $t$ with child nodes $\mswarrow v$ and $\msearrow v$. Then the \emph{type} $\type(v)$ of $v$ is the set of all satisfied propositional formulas with variables of the form $\sigma, \mswarrow \sigma$, and $\msearrow \sigma$, for all $\sigma \in \Sigma$ as well as $s < t$ where $s, t \in \{r, \mswarrow r, \msearrow r \mid  r \in I\}$. Here satisfaction is defined as before. The set of possible types is finite and henceforth denoted~$\Sigma_\types$.

The input tree $t$ can be converted into a \emph{type tree} $\type(t)$ over $\Sigma_\types$ as follows. The tree $\type(t)$ has the same set of nodes as  $t$ and each node $v$ is labelled by $\type(v)$.

\begin{lemma}\label{lemma:ra_type_correspondance}
        One can construct a parity tree automaton $\calA'$ over $\Sigma_\types$ such that for any tree $t$ over $\Sigma \times \data^I$, 
$$\text{$\calA$ accepts $t\quad$ if and only if $\quad \calA'$ accepts $\type(t)$}.$$
\end{lemma}

\begin{proof}This follows from the fact that whether $\calA$ accepts $t$ depends only on the formulas satisfied by  each triple of neighbouring nodes.
\end{proof}

\begin{lemma} \label{lemma:ra_type_consistency}
        One can construct a parity automaton $\calB'$ over $\Sigma_\types$ such that
        for every tree~$t'$ over~$\Sigma_\types$:
        {\small
        $$\text{$\calB'$ accepts $t'$ if and only if $t'=\type(t)$ for some tree $t$ over  $\Sigma \times \data^I$}.$$
        }
\end{lemma}

\begin{proof}[Proof sketch]
  The automaton $\calB'$ only checks for consistency of neighbouring labels in $t'$. Before explaining why this suffices, we define this notion more precisely.

  We say that a triple of labels $\tau, \mswarrow \tau, \msearrow \tau \in \Sigma_\types$ is \emph{consistent} if there exists a finite binary tree $t$ over $\Sigma \times \data^I$ such that in $\type(t)$ the root and its two sons  have labels $\tau, \mswarrow \tau, \msearrow \tau$, respectively. The tree $t$ is called a witness for $\tau, \mswarrow \tau$, and $\msearrow \tau$. 

  Note that it suffices to consider witnesses which are binary trees of height~2. Therefore, testing whether a triple of labels $\tau, \mswarrow \tau, \msearrow \tau$ is consistent amounts to testing satisfiability of a formula over constantly many variables $\Sigma \times \data^I$. It boils down to testing if propositional variables that encode inequalities of register values are consistent. For example, if $\tau$ states that two registers of the right child compare in a certain way, say $\msearrow s < \msearrow t$ for two registers $s$ and $t$, then $\msearrow \tau$ must agree with that, i.e.\ it must state that $s < t$. This discussion implies that set of consistent triples is computable. Obviously, it is also finite, since $\Sigma_\types$ is finite.

  The automaton $\calB'$ checks that for every node $v$ with children $\mswarrow v$ and $\msearrow v$, the triple of labels of these nodes is consistent. 
        
  Clearly, if $t'=\type(t)$ for some tree $t$ over $\Sigma \times \data^I$, then $t'$ is accepted by $\calB'$. It therefore remains to show that if a tree $t'$ is accepted by $\calB'$, then there exists a tree $t$ over $\Sigma \times \data^I$ such that $t=\type(t')$. We construct the tree $t$ in a top-down fashion, defining the labels of the nodes of $t$ by induction on their height. In each step, we use the values provided by the witnesses, and deform them using density of $\data$ so that they fit well into the values of $t$ constructed so far.\end{proof}

\begin{proofof}{Lemma \ref{lemma:emptiness_TRA}}
  From the above two lemmas it follows that $\calA$ is non-empty if and only if there is a tree $t'$ over $\Sigma_\types$ which is accepted both by $\calA'$ and $\calB'$. Indeed, if $t'$ is accepted by $\calA'$ and $\calB'$, then there is a tree $t$ over $\Sigma \times \data^I$ with $t = \type(t')$ by Lemma \ref{lemma:ra_type_consistency} which is accepted by $\calA$ due to Lemma \ref{lemma:ra_type_correspondance}. The other direction is immediate. 

  Since parity automata are closed under intersection, there is a parity tree automaton $\calC'$ which accepts a tree $t'$ if and only if it is accepted both by $\calA'$ and by $\calB'$. Therefore, deciding emptiness of $\calA$ is equivalent to deciding emptiness of $\calC'$. 
\end{proofof}

\section{Satisfiability of Two-Variable Logic with two Orders}\label{section:FO2toRA}
In this section we prove Theorem~\ref{theorem:satisfiability}, that is 
present a decision procedure for two variable logic over countable structures that include one tree order $<_1$, one linear order $<_2$, and access to binary atoms definable in MSO over~$<_1$. Our decision procedure uses the emptiness test for tree register automata with suprema and infima constraints developed in the preceding section. 

The \emph{two-variable fragment of first-order logic} (short: $\fotwo$) restricts first-order logic to the use of only two variables, which can be reused arbitrarily in a formula. The \emph{two-variable fragment of existential second-order logic} (short: $\esotwo$) consists of all formulas of the form $\exists R_1\ldots\exists R_k \varphi$ where $R_1,\ldots,R_k$ are relation variables and $\varphi$ is a $\fotwo$-formula. Since each $\fotwo$-atom can contain at most two variables we assume, from now on and without loss of generality, that all relation symbols are of arity at most two; see \cite[page 5]{GradelKV97} for a justification. Restricting the set quantifiers of $\esotwo$ to be unary yields the fragment $\emsotwo$.

Our interest is in the satisfiability problems for $\esotwo$ and  $\emsotwo$ with two orders $<_1$ and $<_2$. It is easy to see that two-variable logics cannot express that a binary relation is transitive and in particular that it is an order. Therefore we study two-variable logics, where certain relation symbols are supposed to be interpreted only by order relations. More formally, for a class $\calK$ of structures  over some fixed signature $\Delta$, we write $\esotwo(\calK)$ if we are only interested in evaluating formulas in $\esotwo$ with structures from $\calK$. An $\esotwo(\calK)$ formula is satisfiable if it has a model $\frakA \in \calK$, and we say that such a formula is \emph{$\calK$-satisfiable}. Two $\esotwo(\calK)$ formulas are equivalent, if they are satisfied by the same structures from $\calK$.

By $\esotwo(<_1, <_2)$ and  $\emsotwo(<_1, <_2)$ we denote that $\Delta = \{<_1, <_2\}$ where $<_1$ and $<_2$ are interpreted as orders. Often we will be interested in orders of a certain shape only, which we will then state explicitly. For instance, when saying ``$\esotwo(<_1, <_2)$ for a tree order $<_1$ and a linear order $<_2$'', we mean $\esotwo$ formulas with two distinguished binary relation symbols $<_1$ and $<_2$ which are always interpreted by a tree order and a linear order.  For simplicity, we assume that domains always have size at least two and we often identify relation symbols with their respective interpretations.

Let  $\emsotwo(\mso(<_1), <_2)$ denote the set of \mbox{$\emsotwo(<_1, <_2)$} formulas that can use binary atoms definable in $\mso(<_1)$ in addition to atoms from $\{<_1, <_2\}$.
Below, we consider the problem of deciding  whether a given formula of this logic is satisfied in some countable structure equipped with a tree order~$<_1$ and linear order~$<_2$. The following is a restatement of Theorem~\ref{theorem:satisfiability}, and is the main result of the paper.

\begin{theorem}\label{theorem:maintheorem}
   For the logic ${\emsotwo(\mso(<_1), <_2)}$, where~$<_1$ is a tree order and $<_2$ is a linear order, the countable satisfiability problem is decidable.
\end{theorem}

The theorem implies a decision procedure for general satisfiability (i.e., not necessarily restricted to countable domains) for the logic 
${\emsotwo(\fo(<_1), <_2)}$, i.e., for the logic where the additional atoms are restricted to be definable in $\fo$. This follows from the observation that (general) satisfiability of a formula $\exists \bar Z \varphi$, where $\varphi$ is first-order, is equivalent to (general) satisfiability of $\varphi$ (as a formula over extended signature $\{<_1, <_2, \bar Z\}$); and the downward Löwenheim-Skolem theorem for first-order logic.

\begin{corollary} \label{corollary:emsotwo:twolinearorders}
  Satisfiability of $\emsotwo(\fo(<_1), <_2)$ for a tree order $<_1$ and a linear order $<_2$ is decidable.
\end{corollary}

Furthermore, the theorem implies decision procedures for \emsotwo with \mso-definable and \fo-definable classes of orders $<_1$. We give two example corollaries.

\begin{corollary} \label{corollary:emsotwo:twolinearorders'}
   Satisfiability of $\emsotwo(<_1, <_2)$ is decidable when $<_2$ is a linear order and $<_1$ is (i) a linear order, or  (ii) the usual linear order on the naturals, or (iii) the usual linear order on the integers.
\end{corollary}
\begin{proof}For a tree order $<_1$ one can axiomatize in \fo that it is a linear order. Applying Corollary \ref{corollary:emsotwo:twolinearorders} 
  yields statement~(i). The usual order on the naturals and integers can be axiomatised in \mso, and thus statements (ii) and (iii) follow from Theorem \ref{theorem:maintheorem}.
\end{proof}

Using techniques developed by Pratt-Hartmann in \cite{Pratt-Hartmann18} for partial orders (see Section~\ref{section:reduction:esotwo:emsotwo}), one can reduce satisfiability of \mbox{$\esotwo(<_1, <_2)$} to satisfiability of \mbox{$\emsotwo(<_1, <_2)$} when $<_1$ and $<_2$ are linear orders. Combining this reduction with Corollary~\ref{corollary:emsotwo:twolinearorders'} yields the following result. 

\begin{theorem}  \label{thm:esotwo:twolinearorders}
  Satisfiability of $\esotwo(<_1, <_2)$ for linear orders $<_1$ and $<_2$ is decidable.
\end{theorem}

The proof of Theorem \ref{theorem:maintheorem} is provided in Section \ref{section:tree-order-linea-order-decdiability}. The proof of Theorem~\ref{thm:esotwo:twolinearorders} is presented in Section \ref{section:reduction:esotwo:emsotwo}.

\subsection{Decision Procedure for a Tree Order and a Linear Order} \label{section:tree-order-linea-order-decdiability}

The main ingredient for the proof of Theorem \ref{theorem:maintheorem}  is a similar decidability result, for the the special case when $<_1$ is restricted to be the ancestor order of the complete binary tree (see Section \ref{section:AncestorOrderBinaryTree}). 
We then prove that every countable tree order can be interpreted in the ancestor order of the complete binary tree in Section \ref{section:InterpretingCountableOrders}.
Combining the two  yields Theorem \ref{theorem:maintheorem} (see Section~\ref{sec:final-proof}).

\subsubsection{Countable satisfiability for $\emsotwo(\mso(<_1), <_2)$ in the Full Binary Tree} \label{section:AncestorOrderBinaryTree}

In this section the focus is on deciding $\emsotwo(\mso(<_1), <_2)$ in the full binary tree.

\begin{proposition}\label{proposition:emsotwo:binarytree}
  For the logic $\emsotwo(\mso(<_1), <_2)$, where $<_1$ is the ancestor order of the full binary tree and $<_2$ is a linear order, the countable satisfiability problem is decidable.
\end{proposition}

	The high level idea for testing countable satisfiability of an $\emsotwo(\mso(<_1), <_2)$ sentence $\varphi$ is to construct, from $\varphi$, a TRASI $\calA$ such that $\varphi$ has a model if and only if $\calA$ accepts some input tree. Then Proposition \ref{proposition:emsotwo:binarytree} follows immediately as emptiness for TRASIs is decidable. 
  
  Instead of constructing $\calA$ from  a general $\emsotwo(\mso(<_1), <_2)$ sentence $\varphi$, it will be convenient to start from a normal form for such formulas. Say that an  $\emsotwo(\mso(<_1), <_2)$-formula is in \emph{Scott normal form}, if it adheres to the form  
    
    \[\exists \bar Z\Big( \; \forall x \forall y \psi(x,y) \wedge \bigwedge_i \forall x \exists y \psi_i(x, y) \Big)\] 
	where $\psi$ and all $\psi_i$ are quantifier-free formulas whose atoms are $\mso(<_1)$-formulas and relation symbols from $\bar Z \cup \{<_2\}$. It is well known that an $\emsotwo(\mso(<_1), <_2)$-formula can be effectively transformed into an equisatisfiable $\emsotwo(\mso(<_1), <_2)$-formula in Scott normal form (see, e.g., \cite{Scott1962} and \cite[page 17]{GraedelO1999}).

	Under the assumption that $<_2$ is a linear order, formulas in Scott normal form can be further simplified. An $\emsotwo(\mso(<_1), <_2)$-formula is in \emph{constraint normal form} if it is of the form $\exists \bar Z \psi$ where $\psi$ is a conjunction of \emph{existential constraints} of the form \[\forall x \exists y \Big(\bigvee_i \big(x <_2 y \wedge \eta_i(x,y)\big)  \vee \bigvee_i \big(x >_2 y \wedge \vartheta_i(x,y)\big)\Big)\] and of \emph{universal constraints} of the form \[\forall x \forall y \big(x <_2 y \rightarrow  \eta(x,y)\big) \quad \text{ or } \quad \forall x \forall y \big(x = y \rightarrow  \zeta(x)\big)\] where $\eta_i$, $\vartheta_i $, and $\eta$ are arbitrary, quantifier-free $\mso(<_1)$ formulas with free variables $x$ and~$y$ which may use relation symbols from $\bar Z$, and $\zeta$ is a quantifier-free $\mso$ formula with free variable $x$ over $\bar Z$.

	Formulas in Scott normal form can be easily converted into equisatisfiable formulas in constraint normal form by using that $<_2$ is a linear order (this leads to the restricted use of $<_2$-atoms in the constraints).
  \begin{lemma}\label{lemma:constraint-formulas}
		Every $\emsotwo(\mso(<_1), <_2)$-formula can be effectively transformed into an equisatisfiable formula in constraint normal form.
  \end{lemma}
The proof is standard (see, e.g., \cite[Lemma 2.2]{ZeumeH16}) and  therefore omitted here.

For the construction of $\calA$ we will henceforth assume that an \mbox{$\emsotwo(\mso(<_1), <_2)$}-formula in constraint normal form is given. Let $\Theta$ be the signature of $\varphi$ and $\Sigma$ the set of unary, quantifier-free types over~$\Theta$.
  
  The automaton $\calA$ works on input trees over $\Sigma$ with a single data value per node. Intuitively the tree order of such an input tree will encode the order $<_1$ in $\varphi$  and the order on data values will encode the order $<_2$. In this intuition we assume that all data values of data trees are distinct: while the automaton cannot ensure this, we will later see a workaround. The automaton $\calA$ will be constructed such that it accepts a data tree over $\Sigma$ if and only if $\varphi$ is satisfiable. Even more, there will be a one-to-one correspondence (up to morphisms of data values) between data trees $t$ with distinct data values accepted by $\calA$ and models of~$\varphi$. We next outline this correspondence and then show how to construct~$\calA$. 
  
  Each structure $\calS$ over $\Theta$ corresponds to a data tree $t$ over $\Sigma$ with distinct data values as follows: the nodes of $t$ are the elements of the domain of $\calS$ labeled by their unary type in $\calS$; the tree order of $t$ is given by $<_1$; and the order of the data values is given by $<_2$, that is, data values are assigned such that $d(u) < d(v)$ if and only if $u <_2 v$ for all nodes $u$ and~$v$. Here and in the following, $d(v)$ denotes the single data value of a node $v$.  Since $<_2$ is a linear order, all nodes of $t$ have distinct data values. 
  
  On the other hand, a data tree $t$ with distinct data values corresponds to the structure~$\calS$ whose domain contains all nodes of $t$; unary relations are chosen  according to the $\Sigma$-labels of~$t$; and $<_1$ and $<_2$ are interpreted according to the tree order of $t$ and the order of the unique data values, respectively. 
  
  For a data tree $t$ with distinct data values, we say that $t$ \emph{satisfies} $\varphi$ if its corresponding structure $\calS$ satisfies $\varphi$.  

  \begin{lemma}\label{lemma:automaton}
		From an  $\emsotwo(\mso(<_1), <_2)$-formula $\varphi$ in constraint normal form one can effectively construct a TRASI $\calA$  such that
		\begin{enumerate}
		 \item for data trees $t$ with distinct values, $\calA$ accepts $t$ if and only if $t$ satisfies $\varphi$, and 
		 \item if $\calA$ accepts some data tree, then it also accepts one with distinct data values.
		\end{enumerate}
  \end{lemma}

	Thus the TRASI from the lemma ensures that a data tree is accepted if and only if $\varphi$ has a model. Combining this lemma with the decision procedure for emptiness of TRASIs yields Proposition \ref{proposition:emsotwo:binarytree}.

	The TRASI we construct for proving Lemma \ref{lemma:automaton} has to handle $\mso(<_1)$ formulas occuring in constraints. To do so, it will store type-information of nodes. Recall that a $k$-ary \mso-type $\tau(x_1, \ldots, x_k)$ of quantifier-rank $q$ is a maximal consistent conjunction of \mso-formulas of quantifier-rank~$q$ and free variables $x_1, \ldots, x_k$ (see \cite[Section 7.2]{Libkin04} for background on \mso-types). For a tree $t$ with ancestor order $<$ and an $\mso(<)$-type $\tau$, we write $(t, u_1, \ldots, u_k) \models \tau$ if the tuple $(u_1, \ldots, u_k)$ of nodes in $t$ is of type $\tau$. 

	It is folklore that the rank-$q$ \mso-type of a tree can be constructed from the rank-$q$ \mso-types of its components. This is implicit in Shelah's Theorem on \emph{generalised sums} \cite{Shelah75}, see also the exposition from Blumensath et al.~\cite[Theorem 3.16]{BlumensathCL08}.

  We will use the following two special cases stating how the binary type of two nodes $u$ and $v$ in a tree $t$ can be inferred from type information of parts of $t$. The TRASI will store and verify this additional type information and thereby be able to handle $\mso(<_1)$ formulas occuring in constraints. As the MSO formulas in constraints have at most two free variables, it suffices to consider unary and binary types only.  
  
  The first special case states that the type of two nodes $u$ and $v$ in a tree $t$ can be determined from the viewpoint of $v$ by decomposing $t$ into the components $t \setminus t_v$, $t_{\mswarrow v}$,  $t_{\msearrow v}$, and $v$, and combining their respective types. 
  \begin{lemma}\label{lemma:composition}
			Suppose $v$ is a node of a tree $t$. If $u$ is a node in $t \setminus t_v$, then the rank-$q$ \mso-type of $(u, v)$ can be effectively obtained from the rank-$q$ \mso-types of (1) $(\muparrow v, u)$ in $t \setminus t_v$, (2) $\mswarrow v$ in  $t_{\mswarrow v}$, (3)  $\msearrow v$ in  $t_{\msearrow v}$, and of (4) $v$ in the single-node tree consisting solely of~$v$. Similarly if $u$ is a node of $t_{\mswarrow v}$ or $t_{\msearrow v}$, or if $u = v$.
  \end{lemma}

  The second special cases states that one can annotate all nodes of a tree such that this information allows to determine the type of a node $u$ and its ancestor $v$ solely by looking at the annotated path between them. Intuitively, each node is annotated by the information of how it relates to the rest of the tree MSO-wise.
  \begin{lemma}\label{lemma:compositionB}
		For every binary $\mso(<)$-type $\tau(x,y)$ with free first-order variables $x$ and $y$ there is an alphabet $\Gamma$ and a regular language $L$ over $\Gamma$ such that for every tree $t$ with ancestor order $<$ there is a labelling with symbols from $\Gamma$ which is $\mso(<)$-definable and such that for all pairs $(u, v)$ of nodes of $t$ with $u < v$:  $(t, u, v) \models \tau$ if and only if the labels on the path from $u$ to $v$ constitute a string from $L$. 
  \end{lemma}
	
	We are now ready to proceed with the proof of  Lemma \ref{lemma:automaton}.
\begin{proof}[Proof of Lemma \ref{lemma:automaton}]
  
  The automaton $\calA$ is the intersection of two automata $\calA_\varphi$ and $\calA_{\neq}$. Intuitively, the automaton $\calA_\varphi$ verifies that an input tree satisfies $\varphi$ under the assumption that all data values are distinct, but with unspecified behaviour for input trees with distinct nodes with the same data value. The automaton $\calA_{\neq}$ ensures that if a tree is accepted by $\calA_\varphi$ (with possibly the same data value on more than one node) then so is a tree with distinct data values on all nodes.  Note that TRASIs can be easily seen to be closed under intersection. In the following we formalize these intuitions.

  We first explain the construction of the automaton $\calA_\varphi$. This automaton checks whether an input tree  $t$ satisfies each of the existential and universal constraints in $\varphi$. To this end, it will (1) guess and verify $\mso(<_1)$-type information for each node (recall that $<_1$ corresponds to the tree order), in order to handle the $\mso(<_1)$-formulas $\eta_i$, $\vartheta_i$, $\eta$, and $\zeta$ in the constraints, and (2) use additional registers for checking the constraints themselves.

  The automaton  $\calA_\varphi$ guesses and verifies the following type information. Suppose that $q \in \N$ is the maximal quantifier-rank of the formulas $\eta_i$, $\vartheta_i$, $\eta$, and $\zeta$. For each node $v$ the automaton guesses the following unary $\mso(<_1)$-types of quantifier-rank $q$:  
  \begin{itemize}
   \item  the type $\muparrow\delta(v)$ of $\muparrow v$ in $t \setminus t_v$,
   \item  the type $\mswarrow\delta(v)$ of $\mswarrow v$ in $t_{\swarrow v}$, and 
   \item  the type $\msearrow\delta(v)$ of $\msearrow v$ in $t_{\searrow v}$.
  \end{itemize}
  Those guesses can be verified since tree register automata capture the power of parity automata, and therefore also the power of $\mso$.

  The automaton $\calA_\varphi$ uses two sets (A) and (B) of additional registers. The idea is that the registers from (A) will be used to verify that the given data tree satisfies $\varphi$; and registers from (B) will be used to verify the content of registers from~(A).
  
  The set (A) contains registers $r_{\tau, \sup}$ and $r_{\tau, \inf}$ for every binary $\mso(<_1)$-type $\tau(x,y)$ of quantifier-rank at most~$q$. At a node $v$ of an input tree $t$, these registers are intended to store the supremum and infimum data value over all nodes $u$ with $(t, v, u) \models \tau$, that is it is intended that
  \begin{align*}
    \alpha(v, r_{\tau, \sup}) &= \sup \{ d(u) \mid \text{$u \in t$ and $(t, v, u) \models \tau$}\} \\
  \alpha(v, r_{\tau, \inf}) &= \inf \{ d(u) \mid \text{$u \in t$ and $(t, v, u) \models \tau$}\}
  \end{align*}
  for register assignments $\alpha$.

  Assuming, that the information stored in registers in (A) is correct, the automaton can locally check that each node satisfies the existential and universal constraints. For instance, for checking existential constraints, the automaton can guess which of the disjuncts is satisfied for each node of the input tree. If it guesses that for a node $v$ there is a node $u$ such that $v <_2 u \wedge \eta_i(v,u)$, then it can verify this guess by checking whether there is a type $\tau$ that implies $\eta_i$ and a node $u$ such that $d(u) > d(v)$ and $(v, u) \models \tau$. This can be tested via a constraint $r_{\tau, \sup}(v) > d(v)$. The verification of universal constraints is analogous.
  
  The values of the registers in (A) are verified with registers from a set (B) of registers. The intuition for registers in (B) is as follows. For a node~$v$, the supremum of data values of nodes $u$ with $(t, v, u) \models \tau$ is either achieved in $v$ itself, or by nodes in $t \setminus t_v$, or in the left subtree $t_{\mswarrow v}$ of $v$, or in the right subtree $t_{\msearrow v}$, that is:
  {
    \setlength{\tabcolsep}{0pt}
  \[\alpha(v, r_{\tau, \sup}) =  \max 
    \left \{
  \begin{tabular}{rl}
   & $\sup_u\{d(u) \mid u = v \text{ and } (t, v, u) \models \tau \}, $ \\
   & $\sup_u\{d(u) \mid u \in t \setminus t_v  \text{ and } (t, v, u) \models \tau\}, $ \\
   & $\sup_u\{d(u) \mid  u \in t_{\mswarrow v} \text{ and } (t, v, u) \models \tau\}, $ \\
   & $\sup_u\{d(u) \mid  u \in t_{\msearrow v} \text{ and } (t, v, u) \models \tau\}$ \\
  \end{tabular}
\right \}  
  \]}
  
The registers in (B) will allow the automaton to deduce the components of the maximum on the right-hand side.

  The set (B) contains the following registers for each binary type $\gamma(x, y)$ of quantifier-rank~$q$:
    \begin{enumerate}[(i)]
     \item a register $r_{\gamma, \sup, \muparrow}$ intended to store, for a node $v$, the supremum of all nodes $u$ in $t \setminus t_v$ with $(t \setminus t_v, \muparrow v,u) \models \gamma$,
     \item a register $r_{\gamma, \sup, \mswarrow}$ intended to store, for a node $v$, the supremum of all nodes $u$ in $t_{\mswarrow v}$ with $(t_{\mswarrow v}, \mswarrow v,u) \models \gamma$,
     \item a register $r_{\gamma, \sup, \msearrow}$ intended to store, for a node $v$, the supremum of all nodes $u$ in $t_{\msearrow v}$ with $(t_{\msearrow v}, \msearrow v,u) \models \gamma$,
    \end{enumerate}
  Analogously it contains further registers $r_{\gamma, \inf, \muparrow}$, $r_{\gamma, \inf, \mswarrow}$, and $r_{\gamma, \inf, \msearrow}$.
  
  Assuming that the information stored in the register from (B) is correct, the automaton can verify the values of registers in (A) using the intuition stated above, i.e.,  that, for a node~$v$, the supremum of data values of nodes $u$ with $(t, v, u) \models \tau$ is either achieved in $v$ itself, or by nodes in $t \setminus t_v$, or in the left subtree $t_{\mswarrow v}$ of $v$, or in the right subtree $t_{\msearrow v}$.  Thus the values of registers in (A) can be checked as follows:
  
  {\small
  \[\alpha(v, r_{\tau, \sup})  = \max 
    \left \{
  \begin{tabular}{rl}
   & $\{d(v) \mid \text{composing $\muparrow\delta(v)$, $\mswarrow\delta(v)$, $\msearrow\delta(v)$, and the type of $(v,v)$ yields $\tau$}\}$ \\
   $\cup$ & $\{\alpha(v, r_{\gamma, \sup, \muparrow}) \mid \text{composing $\gamma$, $\mswarrow\delta(v)$, $\msearrow\delta(v)$, and the type of $v$ yields $\tau$}\}$ \\
   $\cup$ & $\{\alpha(v, r_{\gamma, \sup, \mswarrow}) \mid \text{composing $\muparrow\delta(v)$, $\gamma$, $\msearrow\delta(v)$, and the type of $v$ yields $\tau$}\}$ \\
   $\cup$ & $\{\alpha(v, r_{\gamma, \sup, \msearrow}) \mid \text{composing $\muparrow\delta(v)$, $\mswarrow\delta(v)$, $\gamma$, and the type of $v$ yields $\tau$}\}$ \\
  \end{tabular}
\right \}  
  \]
  }

  The second line, for instance, takes care of the case when the supremum is achieved in~$t \setminus t_v$. It collects, for each $\gamma$, the suprema values $\alpha(v, r_{\gamma, \sup, \muparrow})$ achieved by nodes $u$ above $v$ with $(t \setminus t_v, \muparrow v, u) \models \gamma$, but only if composing $\gamma$ with the type of the rest of the tree yields the type $\tau$. That is, if the types (1) $\gamma$ of $(\muparrow v, u)$ in $t \setminus t_v$, (2) $\mswarrow\delta(v)$ of $\mswarrow v$ in  $t_{\mswarrow v}$, (3) $\msearrow\delta(v)$ of $\msearrow v$ in $t_{\msearrow v}$, (4) $v$ in the single-node tree consisting solely of~$v$, compose to $\tau$. Here, the composition of types is according to Lemma \ref{lemma:composition}. The equation can be easily implemented by a TRASI.

  We sketch how the automaton can verify the content of the registers in (B). 
  
  For verifying the values of $r_{\tau, \sup, \mswarrow}$ and $r_{\tau, \sup, \msearrow}$ in a register assignment $\alpha$, we use additional registers $r_{\tau, \sup, \mdownarrow}$ with an associated language $L$, such that the supremum of these registers with respect to $L$-descendants is the supremum over all nodes $u$ with $(t_v, v, u) \models \tau$. Thus the values of $r_{\tau, \sup, \mswarrow}$ and $r_{\tau, \sup, \msearrow}$ can be verified by comparing them to these suprema. The register  $r_{\tau, \sup, \mdownarrow}$ stores the input data value of a node $v$. Its associated language $L$ is the language over alphabet $\Gamma$ for $\tau$ according to Lemma \ref{lemma:compositionB}. As TRASI capture the power of MSO, the automaton can guess and verify the $\Gamma$-labeling.

  The value of $r_{\tau, \sup, \muparrow}$ on a node $v$ can be verified as follows. Suppose that $v$ is the left child of its parent $w$ (the cases when $v$ is a right child or the root are very similar). Similarly to the reasoning above, the value of $r_{\tau, \sup, \muparrow}$ is either achieved in $w$, or in $t \setminus t_w$, or in $t_{\msearrow w}$. Thus it can be computed using the following equation:
    
    {\small
    \[\alpha(v, r_{\gamma, \sup, \muparrow})  = \max 
    \left \{
  \begin{tabular}{rl}
   & $\{d(w) \mid \text{composing $\muparrow\delta(w)$, $\msearrow\delta(w)$, and the type of $w$ yields $\gamma$}\}$ \\
   $\cup$ & $\{\alpha(w, r_{\gamma', \sup, \muparrow}) \mid \text{composing $\gamma'$, $\msearrow\delta(w)$, and the type of $w$ yields $\gamma$}\}$ \\
   $\cup$ & $\{\alpha(w, r_{\gamma', \sup, \msearrow}) \mid \text{composing $\muparrow\delta(w)$, $\gamma'$, and the type of $w$ yields $\gamma$}\}$ \\
  \end{tabular}
\right \}  
  \]
  }
  
  This equation can be implemented easily by a tree register automaton with additional auxiliary registers. This concludes the construction of~$\calA_\varphi$.
  
  A tree register automaton with extrema constraints cannot verify that all data values of a tree are distinct. Therefore the automaton $\calA_{\neq}$ checks a weaker condition for a register assignment of $\calA_\varphi$, namely, that if the input value of two nodes $v$ and $v'$ is $d$ then there is a node $u$ on the shortest path from $v$ to $v'$ (which may use tree edges in either direction) such that no register (of $A_\varphi$) in $u$ has value $d$. We call an assignment with this property \emph{weakly diverse}. A careful analysis of the proof of decidability of the emptiness problem for tree register automata with extrema constraints shows that if there is an accepting run which is weakly diverse, then there is also a run with distinct data values. Basically, because the values in a subtree $t_v$ can be shifted by $\pi$ in such a way that fresh data values in registers at node $v$ are distinct from all data values in $t \setminus t_v$.
  
	We show how a TRASI $\calA_{\neq}$ can check that the data values of an input data tree for $\calA_\varphi$ are weakly diverse. For this it has to be verified that no two nodes $v$ and $v'$ with the same  data value $d$ are connected by a path on which each node has a register with value~$d$. In other words, no value $d$ of a register may originate from the input register, hereafter denoted~$s$, of two different nodes $v$ and $v'$. 
  
  Towards checking this condition, $\calA_{\neq}$ guesses a \emph{source} $(v, r)$ for each data value $d$, where $v$ is a node and $r$ is a register. It first ensures that on all nodes the input register is a source. Now, verifying that no data value has two sources ensures that the assignment is weakly diverse. For this verification, sources are annotated by labels from $\{S_r \mid  \text{$r$ is a register}\}$ with the intention that a node $u$ is labelled by $S_r$ if $(u, r)$ is a source. For tracing values to their source, each node is annotated by labels from $\{S_{r}^\uparrow, S_{r}^\swarrow, S_{r}^\searrow \mid  \text{$r$ is a register}\}$ with the intention that, e.g., if a node $u$ is labelled by $S_{r}^\swarrow$ then the value of register $r$ at node $u$ originates somewhere in the left subtree of $u$. 

  The automaton then checks the following for each node~$u$:
  \begin{enumerate}
    \item The node $u$ is labelled by exactly one of the labels $S_r, S_{r}^\uparrow, S_{r}^\swarrow$, or $S_{r}^\searrow$, for each register~$r$. Furthermore the labels are consistent for registers with the same value, i.e. if \mbox{$\alpha(u, r) = \alpha(u, r')$} then $u$ is labelled by $S_r$ if and only if it is labeled by $S_{r'}$, and similarly for $S_{r}^\uparrow$, $S_{r}^\swarrow$, and $S_{r}^\searrow$.
    \item The labels $S_r$, $S_{r}^\uparrow, S_{r}^\swarrow$, or $S_{r}^\searrow$ indicate the direction of the source of the data value of register $r$ at node $u$:
      \begin{enumerate}
        \item If $u$ is the left child of its parent $v$ and $\alpha(u, r) = \alpha(v, r')$, then
          \begin{itemize}
            \item If $u$ is labelled by $S_r$, $S_{r}^\searrow$ or $S_{r}^\swarrow$ then $v$ is labelled by $S_{r'}^\swarrow$.
            \item If $u$ is labelled by $S_{r}^\uparrow$ then $v$ is labelled by $S_{r'}, S_{r'}^\searrow$, or $S_{r'}^\uparrow$.
          \end{itemize}
        \item Similarly if $u$ is the right child of $v$, and if $u$ is the parent of $v$. 
      \end{enumerate}
    \item The register $s$ of node $u$ is a source, i.e. each node $u$ is labelled by $S_s$. (Recall that $s$ is the distinguished input register.) 
  \end{enumerate}
  The first condition ensures that no data value originates from register $s$ of two different nodes. The automaton $\calA_{\neq}$ now guesses and verifies the described annotation witnessing weak diversity for the input register $s$ with respect to the registers of $\calA_\varphi$. This proves Lemma~\ref{lemma:automaton}.
\end{proof}
Lemma~\ref{lemma:automaton} together with Theorem~\ref{thm:decidable}
yield Proposition~\ref{proposition:emsotwo:binarytree}.

\subsubsection{All Tree Orders Interpret in the Complete Binary Tree} \label{section:InterpretingCountableOrders}
In this section we prove that every countable tree can be encoded in the full binary tree using formulas. This will allow us to derive 
Theorem~\ref{theorem:maintheorem} from Proposition~\ref{proposition:emsotwo:binarytree} in Section~\ref{sec:final-proof}.

A tree order is \emph{universal} if it contains every countable tree order as a substructure.  Denote by $\calT$ the complete binary tree encoded as a structure with domain $V=\set{0,1}^*$ over the signature $\{L, R, <\}$ which interprets $L$ and $R$ as the left and right child relations, and $<$ as the ancestor order of the complete binary tree. 

We strongly suspect that the following result is known; for lack of reference we provide a proof sketch. 
\begin{lemma}\label{lemma:interpretation}
  There are first-order formulas $\delta_U(x)$ and $\delta_\prec(x,y)$ over signature $\{L, R, <\}$ such that the structure $(U, \prec)$ with
  \begin{align*}
    U &\df \{a \in V \mid \calT \models \delta_U(a)\}\\ 
    \prec &\df \{(a, b) \in U \times U \mid \calT \models \delta_\prec(a,b)\}
  \end{align*}
  is a universal tree order.

\end{lemma}

We proceed with the proof of Lemma \ref{lemma:interpretation}. The following notation and facts will be used in the proof. For a tree order $\prec$, we say that $x$ is an \emph{ancestor} of $y$, or that $y$ is a \emph{descendant} of $x$, if $x \prec y$ or $x = y$, written~$x\preccurlyeq y$.
The \emph{meet} $x \wedge y$ of $x$ and $y$ 
is the greatest lower bound of $x$ and $y$, with respect to~$\preccurlyeq$. Observe that if $x, y_1, \ldots, y_m$ are elements of the tree order, then $x \wedge y_1, \ldots, x \wedge y_m$ are ancestors of $x$ (if they exist). In particular, due to the definition of tree orders, the tree order $\prec$ is total on $x \wedge y_1, \ldots, x \wedge y_m$. 

\begin{proof}[Proof of Lemma~\ref{lemma:interpretation}]Instead of working with universality directly, we use  the following stronger condition. A nonempty tree order $\prec$ satisfies the \emph{extension axioms} if:
  \begin{enumerate}
    \item for any node $x$ and nodes $y_1, \ldots, y_n$ which are mutually incomparable descendants of $x$, there exists yet another incomparable descendant $y$,
     \item for every node $x$ there is some $y$ such that $y \prec x$, and     
    \item for any two nodes $x$ and $y$ such that $x \prec y$, there is a node $z$ such that $x\prec z \prec y$.
  \end{enumerate}

  Let us first show that every tree order $\prec$ on a set of nodes $V$ satisfying the extension axioms  is universal. Let $\prec'$ be another tree order with countable domain~$V'$.
  Extending this tree order by countably many points if necessary, we may assume that $\prec'$ is closed under meets, i.e., the meet of any two elements of $V'$ is in $V'$. Fix an enumeration of the elements of the domain $V'$.

  We inductively construct an embedding $f\from V' \to V$ by defining $f$ for larger and larger finite meet-closed subsets $X$ of $V'$. For the induction basis, the first element of $V'$ with respect to the chosen enumeration is mapped to an arbitrary element of $V$. Now suppose that we have already embedded a finite subset $X$ of $V'$ into $V$. In the inductive step we extend the embedding by considering the first node $y$ of $V'-X$ with respect to the chosen enumeration. Consider the set \[M\df\set{y\land x\mid x\in X}\subset V'\] of meets of $y$ with the previously inserted nodes (see Figure \ref{figure:universal}).
  This set is a finite set of ancestors of $y$, hence it has a largest element $m$ with respect to $\prec'$. We observe that the set $M \setminus \{m\}$ is contained in $X$ due to the meet closure of $X$. 
  
  If the element $m$ does not already belong to $X$ then  we first extend $f$ to $X\cup \set m$. To this end let $b$ be the image under $f$ of the $\prec'$-smallest element of the set $\set{x\in X\mid x\succ' m}$, which exists since $X$ is meet-closed.
  
  If $|M| = 1$, apply the second extension axiom to $b$ in order to obtain an element $a \in V$ with $a \prec b$, and define $f(m) \df a$.

  If $|M|>1$ then let $b'$ be the image under $f$ of the second-largest element of $M$ (the second largest element must be in $X$, due to meet-closure). Apply the third  extension axiom to $b'$ and $b$ in order to obtain an element $c$ of $V'$ with $b' \prec c \prec b$ to which $m$ will now be mapped.

  Once $f$ is defined on $X\cup \set m$, use the first extension axiom to extend the embedding to $y$. Finally, replace $X$ by $X\cup \set{m,y}$ and proceed to the next element of $V'-X$. Note that the set $X\cup \set{m,y}$ is closed under meet.

This concludes the proof that $\prec$ is universal.

\begin{figure}
\begin{tikzpicture}[edge from parent/.style={draw, decorate,decoration=snake}]

\node {$x_1 \wedge x_2 \wedge x_3$}
    child { node {$m$} 
      child {node {$y$}}
      child { node {$x_1 \wedge x_2$}
        child { node {$x_1$} }
        child { node {$x_2$} }
       }
    }
    child { node {$x_3$} };
\end{tikzpicture}
\caption{Illustration of the construction of the mapping from $V'$ to $V$ in the proof of Lemma \ref{lemma:interpretation}. Suppose the elements $x_1, x_2, x_3$ are already in $X$. Due to the closure under meets, also $x_1 \wedge x_2$ and $x_1 \wedge x_2 \wedge x_3$ are in $X$. The set $M $ contains the nodes $m = y \wedge x_2$ and nodes $y \wedge x_3$.} \label{figure:universal}

\end{figure}
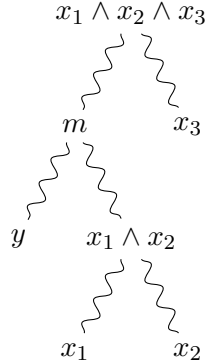

\medskip

  We will now show that a tree order satisfying the extension axioms can be first-order interpreted in the complete binary tree $\cal T$.

  \paragraph*{Step 1.} The complete ranked ternary tree (each node has exactly one ``left child'', one ``right child'' and one ``middle child'') $\cal T'$ and its tree order can be first-order interpreted into the complete binary tree. This can be done by using the unary predicate $V'\subset V=\set{0,1}^*$ 
  defined by the regular expression $(0100+0110+0111)^*$.
  The regular word language $V'\subset\set{0,1}^*$
  can be defined by a first-order sentence $\varphi$ in the signature of 
  words, that is $\set{0,1,<}$, where $<$ is the usual order on the positions of a word  (equivalently, by a star-free expression, by the McNaughton-Papert theorem). 
  The sentence $\varphi$ expresses the following:
  \begin{itemize}
    \item the word starts with $01$,
    \item every occurrence of $01$ is immediately followed either by $0001$, or by $1001$, or by $1101$, or occurs at the end of the word.
  \end{itemize} 
The sentence $\varphi$ translates into a first-order formula 
$\delta'(x)$ that holds at a node $v\in V$ of $\cal T$ if and only of $v\in V'$.
  
Thus, the set of nodes of the ternary tree $\cal T'$ is defined as the set $V'$ of those nodes of $\cal T$ that satisfy $\delta'(x)$. The ancestor relation of the ternary tree $\cal T'$
is then the restriction of the ancestor relation of $\cal T$ to $V'$.

  \paragraph*{Step 2.} The \emph{$\Q$-tree} is the rooted tree $\cal T''$ where each non-root node has a parent,  and the children of each node are ordered ``from left to right'' forming a dense order without endpoints. The $\Q$-tree, its tree order, and the order on siblings interpret in the complete ranked ternary tree~$\cal T'$.

  As a warm-up we observe that the structure $(\Q, \leq)$ interprets in the complete ranked ternary tree $\cal T'$. The elements of $(\Q,\leq)$ correspond to the nodes of $\cal T'$ which are middle children, and which have no ancestor which is a middle child. The order $\leq$ is given by $v \leq w$ if and only if $v$ is the left descendant of the meet of $v$ and $w$.

  For interpreting the $\Q$-tree $\cal T''$ in the complete ranked ternary tree $\cal T'$, the construction for $(\Q, \leq)$ is nested. More explicitly, the nodes of the $\Q$-tree $\cal T''$ are the nodes of $\cal T'$ which are middle children (corresponding to the nodes $(0100+0110+0111)^*0110\subset V$ of $\cal T$), and the ancestor relation is the ancestor relation of $\cal T'$ restricted to those nodes. The left-to-right order on the children $v$, $w$ (in the $\Q$-tree) of a given node is such that $v$ is to the left of $w$ if and only if $v$ is the left descendant of the meet of $v$ and $w$ in $\cal T'$.
  Note that the nodes of $\cal T''$ are again definable by a first-order formula $\delta''(x)$ in $\cal T$, and so is the ancestor relation of $\cal T''$.

  \paragraph*{Step 3.} Consider the interpretation of a tree order $\prec$ in the $\Q$-tree $\cal T''$, in which $v \prec w$ if and only if $v$ is a left sibling of an ancestor of $w$ in the $\Q$-tree. It is easy to see that $\prec$ is indeed a tree order. Also all meets exist in $\prec$, since, for all elements $x$ and $y$, there is a meet $z$ of $x$ and $y$ in the $\Q$-tree with children $x'$ and $y'$ which are ancestors of $x$ and $y$. The smaller of the two nodes $x'$ and $y'$ with respect to the children order of the $\Q$-tree is the meet of $x$ and~$y$.

  We further claim that $\prec$ satisfies the extension axioms, apart from the second one. To show the first extension axiom, suppose that $y_1, \ldots, y_n$ are pairwise incomparable with respect to $\prec$ and such that $x \prec y_i$ for all $i$. Since $x \prec y_i$ for all $i$ there must be a node with children $x_1, \ldots, x_n$ such that each $x_i$ is an ancestor of $y_i$ in the $\Q$-tree $\cal T''$, and without loss of generality the $x_i$ appear in the order $x_1, \ldots, x_n$. Since siblings in the $\Q$-tree are dense, there is a sibling $x'$ to the left of $x_1$ such that $x \prec x'$. Every child $y$ of $x$ in the $\Q$-tree is incomparable to $y_1, \ldots, y_n$ with respect to $\prec$. 

  Towards showing that $\prec$ satisfies the third  extension axiom, suppose that $x \prec y$. If $x$ is not a left sibling of $y$ then choose some left sibling $z$ of $y$, otherwise choose an arbitrary $z$ in-between $x$ and $y$ in the $\Q$-tree (such a $z$ exists as siblings in the $\Q$-tree are dense). In both cases $x \prec z \prec y$. 

The order $\prec$ does not satisfy the second extension axiom, as it contains a smallest element, corresponding to the root of the $\Q$-tree.
However, it follows from the third extension axiom that after removing the root, the second extension axiom is satisfied. Therefore, $\prec$ restricted to the set of non-root nodes of the $\Q$-tree is a universal tree order.

  It is not hard to verify that all steps are indeed definable in first-order logic,
  yielding the required first-order formulas $\delta_U(x)$ and $\delta_\prec(x,y)$.
  More explicitly, $\delta_U(x)$ is the first-order formula that checks membership of $x\in\set{0,1}^*$ in the regular language $U:=(0100+0110+0111)^*0110$,
  while $\delta_\prec(x,y)$ is the first-order formula that holds of $u,v\in U$ if $u$ is lexicographically smaller than $v$ and, 
  writing $u$ as $u=u_1\cdot u_2\cdot 0110$, where $u_1$ is the longest common prefix of $u$ and $v$ belonging to $(0100+0110+0111)^*$, we have that $u_2\in (0100+0111)^*$. It follows from the reasoning presented above (and can be also verified directly), the relation defined by $\delta_\prec(x,y)$ on $U$ is a universal tree order. Moreover, both $\delta_U(x)$ and $\delta_\prec(x,y)$ are first-order definable. This proves Lemma~\ref{lemma:interpretation}.
\end{proof}

\subsubsection{Countable satisfiability for  $\emsotwo(\mso(<_1), <_2)$}\label{sec:final-proof}
We can at last prove Theorem \ref{theorem:maintheorem}.

\begin{proof}[Proof of Theorem \ref{theorem:maintheorem}.]
  Suppose $\varphi$ is an $\emsotwo(\mso(<_1), <_2)$ formula for a tree order $<_1$ and a linear order $<_2$. Then satisfiability of $\varphi$ reduces to the $\emsotwo(\mso(<'_1), <_2)$ formula $\varphi' = \exists S \psi$ where $<'_1$ is the ancestor order of the full binary tree and $\psi$ is obtained from $\varphi$ by 
  \begin{enumerate}[(1)]
   \item relativizing all first-order quantifiers via replacing subformulas $\exists x \gamma(x)$ and $\forall x \gamma(x)$ by $\exists x (S(x) \wedge \delta_{U}(x) \wedge \gamma(x))$ and $\forall x (S(x) \wedge \delta_{U}(x)\rightarrow \gamma(x))$, respectively, and
		\item replacing all atoms $x <_1 y$ by $\delta_\prec(x, y)$. 
  \end{enumerate}
		Here, $\delta_{U}(x)$ and  $\delta_\prec(x, y)$ are the formulas from Lemma \ref{lemma:interpretation}.
		
	Thus, a tree order $<_1$ in a model of $\varphi$ is simulated in models of $\varphi'$ by guessing a substructure (of domain $S$) of the universal tree order defined by $\delta_{U}(x)$ and $\delta_\prec(x,y)$.
\end{proof}

\subsection{Satisfiability of $\esotwo(<_1, <_2)$ reduces to $\emsotwo(<_1, <_2)$}
\label{section:reduction:esotwo:emsotwo}
In this section we prove Theorem~\ref{thm:esotwo:twolinearorders}.

For many classes $\calK$ of structures the general satisfiability problem of $\esotwo(\calK)$ can be reduced to the satisfiability problem of $\emsotwo(\calK)$. Examples are two-variable logic with one linear order \cite{Otto01} and two-variable logic with one relation to be interpreted by a preorder \cite{Pratt-Hartmann18}. 

Here we show that the same approach works for ordered structures with two linear orders. This proves Theorem~\ref{thm:esotwo:twolinearorders}.

\begin{proposition}\label{proposition:binaryToUnary}
  The satisfiability problem of $\esotwo(<_1, <_2)$ effectively reduces to the satisfiability problem  of ${\emsotwo(<_1, <_2)}$, for linear orders $<_1$ and $<_2$.
\end{proposition}

For the proof of this proposition, we rely on a general reduction for a broad range of classes of structures which is implicit in the procedure used by Pratt-Hartmann in \cite{Pratt-Hartmann18}. We make this approach explicit and show that it can also be applied to structures with two linear orders. 

Observe that the satisfiability problems for $\esotwo(\calK)$- and $\emsotwo(\calK)$-formulas are inter-reducible with the satisfiability problems for $\fotwo(\calK)$ over arbitrary signatures and, respectively, signatures restricted to unary relation symbols apart from the signature of $\calK$. For this reason we restrict our attention to two-variable \emph{first-order logic} in the following. One might wonder, why not concentrate on $\fotwo$ throughout, yet we prefer to state results such as Proposition \ref{proposition:binaryToUnary} in terms of \esotwo and \emsotwo instead of \fotwo as the abbreviation \fotwo is ambiguously used for two-variable logic over unary and arbitrary signatures in the literature.

The high level proof idea is as follows. Recall that a $\fotwo$-formula can be transformed in a satisfiability-preserving way into Scott normal form. For structures with at least two elements, formulas in Scott normal form can be transformed into an equisatisfiable formula in \emph{standard normal form}  \cite{Pratt-Hartmann18}, that is, into the form

  \[\forall x \forall y (x = y \vee \psi(x,y)) \wedge \bigwedge^{m}_{h=1} \forall x \exists y (x \neq y \wedge \psi_h(x, y))\]

where both $\psi$ and all $\psi_h$ are quantifier-free. Thus, a formula in standard normal form poses some universal as well as some existential constraints on models. Each existential constraint $\forall x \exists y (x \neq y \wedge \psi_h(x, y))$ requires that all instantiations $a$ of $x$ are witnessed by an instantiation $b$ of $y$ such that $a$ and $b$ satisfy $\psi_h$. Given a model, a \emph{witness function} $W$ assigns a witness $b$ to each $\forall x \exists y (x \neq y \wedge \psi_h(x, y))$ and each instantiation $a$ of $x$. We call such an instantiation $b$ a \emph{witness} for $a$ and $\forall x \exists y (x \neq y \wedge \psi_h(x, y))$.

We will show that (1) under certain conditions, the $\forall \exists$-constraints can be required to \emph{spread} their witnesses over distinct elements (see Lemma \ref{lemma:cloningImpliesSpreading}), and (2) if witnesses are spread in this way, then binary relation symbols can be removed from formulas in a satisfiability-preserving way (see Lemma \ref{lemma:spreadImpliesReduction}). Proposition \ref{proposition:binaryToUnary} will then follow from the observation that $\fotwo(<_1, <_2)$ allows to spread witnesses (see Lemma \ref{lemma:cloningOfTwoLinearOrders}). 

Let us first formalize what we mean by ``spreading''. A witness function $W$ \emph{spreads} its witnesses if for all elements $a$: 
\begin{enumerate}
 \item The witnesses for $a$ for all existential constraints are distinct elements.
 \item If $a$ has a witness $b$ for some constraints, then all witnesses for $b$ are not the element~$a$.
\end{enumerate}

We start by refining the standard normal form to spread normal form with the intention that if a formula in spread normal form has a model, then it has a witness function that spreads its witnesses. Fix a class $\calK$ of structures over some signature $\Delta$. An  $\fotwo(\calK)$ formula over $\Delta \uplus \Theta$ is in \emph{spread normal form} \cite[p. 23]{Pratt-Hartmann18} if it conforms to the pattern
\begin{align*}
  & \quad\bigwedge_{\theta \in Z} \exists x \theta(x) \wedge \forall x \forall y (x = y \vee \nu(x,y)) \\ 
  \wedge &\quad \bigwedge_{k=0}^2 \bigwedge_{h = 0}^{m-1} \forall x \exists y (\lambda_k(x) \rightarrow (\lambda_{\lfloor k+1 \rfloor}(y) \wedge \mu_h(y)\wedge \theta_h(x,y)))
\end{align*}

where (i) $Z$ is a set of unary pure Boolean formulas; (ii) $\nu$, $\theta_0, \ldots, \theta_{m-1}$ are quantifier- and equality-free formulas; (iii) $\lambda_0$, $\lambda_1$, and $\lambda_2$ are mutually exclusive unary pure Boolean formulas;
and (iv) $\mu_0, \ldots, \mu_{m-1}$ are mutually exclusive unary pure Boolean formulas. Here, a formula is said to be \emph{pure Boolean} if it is quantifier-free and only uses symbols from $\Theta$. By $\lfloor k+1 \rfloor$ we denote $k+1 \pmod 3$.  

It is easy to see that for each model of a formula in spread normal there is a witness function that spreads its witnesses. The first condition -- that witnesses for an element $a$ for the existential conjuncts are distinct elements -- is ensured by the mutually exclusive~$\mu_h$. The second condition -- that if $a$ has a witness $b$ then all witnesses for $b$ are not the element~$a$ -- is ensured by the mutually exclusive $\lambda_k$. 
  
Our next goal is to find a general criterion for when $\fotwo(\calK)$ formulas can be transformed into equisatisfiable formulas in spread normal form, for a given class $\calK$ of structures of signature $\Delta$. 
Intuitively, $\fotwo(\calK)$ formulas $\varphi$ can be transformed into a equisatisfiable formula $\varphi^*$ in  spread normal form, if elements can be ``cloned'' without affecting the truth of the formula. Indeed, if elements can be cloned, then disjoint witnesses for existential constraints can be created on-demand by cloning elements. In a moment, after formalizing what it means that an element can be cloned,  we will see that not all elements can be cloned. This will lead us to slightly adapt our goal and to rather look for a criterion for when $\fotwo(\calK)$ formulas can be transformed into a \emph{disjunction} of  equisatisfiable formulas in spread normal form, where each disjunct essentially takes care of one configuration of the non-clonable elements. 

Let us formalize what it means that elements can be cloned. We first fix some notation. In the following we denote the quantifier-free type of a tuple $\bar a$ in a structure $\calA$ by $\type^{\calA}[\bar a]$. We say that a quantifier-free type $\tau$ is \emph{realized} in $\calA$, if there is a tuple $\bar a$ with $\type^{\calA}[\bar a] = \tau$. A structure $\calA$ over signature $\Delta \uplus \Theta$ is an expansion of a $\Delta$-structure $\calB$, if $\calA$ and $\calB$ have the same domain and agree on the interpretations of relation symbols from $\Delta$.

Now, suppose $\calK$ is a class of structures over $\Delta$, and $\calA$ is an expansion of a $\calK$-structure with domain $A$ over a signature $\Delta \uplus \Theta$.  Then $\calA$ \emph{allows cloning} of a set of elements~$B$, if there is a structure $\calA'$ with domain $A' \df A \cup B'$ where $B' = \{b' \mid b \in B\}$ such that
  \begin{enumerate}[(i)]  
    \item $\calA$ is a substructure of $\calA'$ and all quantifier-free binary types realized in $\calA'$ are realized in $\calA$;
    \item $\type^{\calA'}[a, b'] = \type^{\calA}[a, b]$ for all $a, b \in A$ with $a \neq b$; and
    \item $\type^{\calA'}[b'_1, b'_2] = \type^{\calA}[b_1, b_2]$ for all $b_1, b_2 \in A$ with $b_1 \neq b_2$; and    
    \item the $\Delta$-structure obtained from $\cal A'$ by forgetting the relations in $\Theta$  is a $\calK$-structure. 
  \end{enumerate}

It is easy to see that it is not possible to clone elements whose unary quantifier-free type is unique within a structure: when we clone such an element $b$, then the binary type of $b$ and its clone $b'$ has not been present before, which violates condition (i) from above.  This motivates the following definition. An element $a$ of a structure $\calA$ is a \emph{King}, if there is no other element $b$ in $\calA$ with the same quantifier-free type, i.e., there is no $b \neq a$ with $\type^{\calA}[a] = \type^{\calA}[b]$.

Now, we say that a \emph{class $\calK$ of structures allows cloning} if every expansion of a $\calK$-structure allows cloning of arbitrary sets of non-King elements. The following lemma establishes that  in structures with two linear orders we can clone arbitrary sets of non-King elements. The same construction can be used for structures with more than two linear orders.

\begin{lemma}\label{lemma:cloningOfTwoLinearOrders}
  The class of structures with two linear orders allows cloning. 
\end{lemma}
\begin{proof}[Proof sketch]
  Consider a structure $\calA$ with signature  $\{<_1, <_2\} \uplus \Theta$ and domain $A$ that interprets $<_1$ and $<_2$ by linear orders. Let $B$ be a set of non-King elements of $\calA$. The structure $\calA'$ that clones $B$ in $\calA$ has domain $A' \df A \cup \{b' \mid b \in B\}$. The relations between elements in $A$ within the structure $\calA'$ are inherited from $\calA$.   The simple idea for new elements $b'$ is to insert a copy $b'$ of $b \in B$ immediately before/after $b$ with respect to $<_1$ and $<_2$. If $c \in A$ is a non-king with the same unary type as~$b$ (such an element exists since $b$ is a non-King) and if, for instance, $b <_1 c$ and $b <_2 c$ then $b'$ is inserted directly after $b$ in both orders $<_1$ and $<_2$. The type of $(b, b')$ is then inherited from $(b, c)$. Furthermore the element $b'$ inherits its relation to all other elements from~$b'$. See Figure \ref{figure:cloning} for an illustration.
  
  More formally, the tuples of $\calA'$ have the following types :
  \begin{enumerate}[(i)]
    \item $\type^{\calA'}[a] = \type^{\calA}[a]$ and $\type^{\calA'}[a, a'] = \type^{\calA}[a, a']$ for all $a, a' \in A$;
    \item $\type^{\calA'}[b'] = \type^{\calA}[b]$ and $\type^{\calA'}[b', a] = \type^{\calA}[b, a]$ for all $b \in B$ and $a \in A$;
    \item $\type^{\calA'}[b'_1, b'_2] = \type^{\calA'}[b_1, b_2]$ for all $b_1, b_2 \in B$ with $b_1 \neq b_2$;    
    \item $\type^{\calA'}[b, b'] = \type^{\calA}[b, c]$ for $c \in A$ with $\type^{\calA}(c) = \type^{\calA}(b)$.
  \end{enumerate}
  It can be easily verified that $\calA'$ satisfies the conditions on cloning. 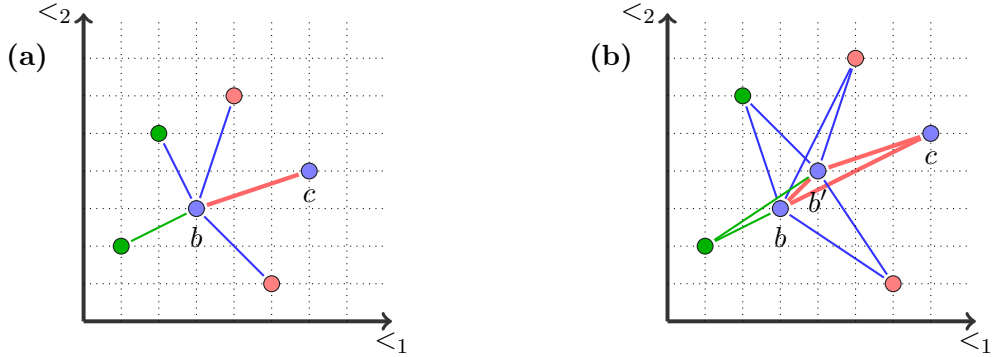
\begin{figure}[t] 
\begin{subfigure}[t]{0.5\textwidth}
      \begin{tikzpicture}[
        xscale=0.5,
        yscale=0.5
      ]
        \node (tmp) at (-1.5, 7) {\textbf{(a)}};
        
\draw [->, line width=1.5pt, black!80] (0,0) -- (8.2,0) node[black, below] {$<_1$};
        \draw [->, line width=1.5pt, black!80] (0,0) -- (0,8.2) node[black, left, sloped] {$<_2$};

\foreach \x in {1,2,...,7}
        \draw[dotted, black] (\x,0) -- (\x,8) ;

        \foreach \y in {1,2,...,7}
        \draw[dotted, black] (0,\y) -- (8,\y);

\node (v1) at (5, 1)[redNode] {};
        \node (v2) at (2, 5)[greenNode] {};
        \node (v3) at (3, 3)[blueNode, label={below:$b$}] {};
        \node (v4) at (4, 6)[redNode] {};
        \node (v5) at (1, 2)[greenNode] {};
        \node (v6) at (6, 4)[blueNode, label={below:$c$}] {};

\draw[blueEdge](v3)  -- (v1);
        \draw[blueEdge](v3)  -- (v2);
        \draw[blueEdge](v3)  -- (v4);
        \draw[greenEdge](v3)  -- (v5);
        \draw[redEdge, ultra thick](v3)  -- (v6);
        
      \end{tikzpicture}
    \end{subfigure}
    \begin{subfigure}[t]{0.45\textwidth}
      \begin{tikzpicture}[
        xscale=0.5,
        yscale=0.5
      ]
        \node (tmp) at (-1.5, 7) {\textbf{(b)}};
        
\draw [->, line width=1.5pt, black!80] (0,0) -- (8.2,0) node[black, below] {$<_1$};
        \draw [->, line width=1.5pt, black!80] (0,0) -- (0,8.2) node[black, left, sloped] {$<_2$};

\foreach \x in {1,2,...,7}
        \draw[dotted, black] (\x,0) -- (\x,8) ;

        \foreach \y in {1,2,...,7}
        \draw[dotted, black] (0,\y) -- (8,\y);

\node (v1) at (6, 1)[redNode] {};
        \node (v2) at (2, 6)[greenNode] {};
        \node (v3) at (3, 3)[blueNode, label={below:$b$}] {};
        \node (v3b) at (4, 4)[blueNode, label={below:$b'$}] {};
        \node (v4) at (5, 7)[redNode] {};
        \node (v5) at (1, 2)[greenNode] {};
        \node (v6) at (7, 5)[blueNode, label={below:$c$}] {};

        \begin{pgfonlayer}{background}
\draw[redEdge, ultra thick](v3)  -- (v3b);

          \draw[blueEdge](v3)  -- (v1);
          \draw[blueEdge](v3)  -- (v2);
          \draw[blueEdge](v3)  -- (v4);
          \draw[greenEdge](v3)  -- (v5);
          \draw[redEdge, ultra thick](v3)  -- (v6);

          \draw[blueEdge](v3b)  -- (v1);
          \draw[blueEdge](v3b)  -- (v2);
          \draw[blueEdge](v3b)  -- (v4);
          \draw[greenEdge](v3b)  -- (v5);
          \draw[redEdge, ultra thick](v3b)  -- (v6);
        \end{pgfonlayer}
    \end{tikzpicture}
    \end{subfigure}
\caption{\textbf{(a)} A structure with non-King elements $b$ and $c$ with the same unary type. \textbf{(b)} The structure obtained by cloning the element $b$ once. For clarity only the binary types between $b$ and its clone $b'$ with respect to other elements are indicated. }\label{figure:cloning}
\end{figure}
\end{proof}

When a class $\calK$ allows cloning then witnesses for a formula $\varphi \in \fotwo(\calK)$ in Scott normal form can be spread among cloned elements. Unfortunately, Kings need to be dealt with separately as they cannot be cloned. But, essentially, for each  possible ``configuration'' of Kings and their witnesses one can find a formula in spread normal form such that their disjunction is equisatisfiable to $\varphi$.

\begin{lemma}[Adaption of Lemma 4.2 in \cite{Pratt-Hartmann18}] \label{lemma:cloningImpliesSpreading}
  If a class $\calK$ of structures allows cloning then for all $\varphi \in \fotwo(\calK)$ one can effectively find formulas $(\varphi_\calC)_{\calC \in C}$, for some finite set $C$, in spread normal form such that the following conditions are equivalent:
	\begin{enumerate}
	 \item $\varphi$ is (finitely) satisfiable.
	 \item $\bigvee_{\calC \in C} \varphi_\calC$ is (finitely) satisfiable.
	\end{enumerate}
\end{lemma}
\begin{proof}This follows immediately from the proof of Lemma 4.2 in \cite{Pratt-Hartmann18}. The structures $\calA_1, \ldots, \calA_i$ used in that proof can be easily constructed inductively by successively cloning non-King elements. The rest of the proof is independent of the class $\calK$. 
  
  For the sake of completeness we provide a proof sketch. Suppose $\calK$ contains structures over signature $\Delta$ and $\varphi$ is a $\fotwo(\calK)$ formula over signature $\Delta \uplus \Theta$ in standard normal form: 
  
  \[\forall x \forall y (x = y \vee \psi(x,y)) \wedge \bigwedge^{m}_{h=1} \forall x \exists y (x \neq y \wedge \psi_h(x, y)\]

	Each of the formulas $\varphi_\calC$ in spread normal form that we will construct is responsible for one configuration of Kings and their witnesses in models of $\varphi$. More precisely, let $\Gamma_K$ be a set of unary types over the signature $\Delta \uplus \Theta$. A \emph{court structure for $\Gamma_K$ and $\varphi$} is a structure $\calC$ such that
   \begin{enumerate}
    \item the types of Kings in $\calC$ are exactly the types in $\Gamma_K$,
    \item $\calC$ satisfies $\forall x \forall y (x = y \vee \psi(x,y))$, and
    \item for all Kings $a$ of $\calC$ and all $h \in \{1, \ldots, m\}$ there is exactly one $b \neq a$ in $\calC$ such that $C \models \psi_h(a, b)$.
   \end{enumerate}
	Thus, such a court structure represents a set of Kings as well as their witnesses. Denote by $C$  the set of all court structures for all sets $\Gamma_K$ of unary types. 
  
  For each $\calA$ with $\calA \models \varphi$, the \emph{court structure induced by $\calA$} is the substructure of $\calA$ induced by the Kings of $\calA$ as well as their witnesses. 
  
  Our goal is to construct formulas $\varphi_\calC$ in spread normal form for each $\calC \in C$ such that the following conditions are equivalent:
		\begin{enumerate}
		 \item $\varphi$ has a (finite) model with court structure $\calC$.
		 \item $\varphi_\calC$ has a (finite) model. 
		\end{enumerate}
	Then the statement of the proposition follows by observing that $\varphi$ has a (finite) model if and only if $\bigvee_{\calC \in C} \varphi_\calC$ has a (finite) model. 
	   
  Thus we aim at constructing, for each $\calC \in C$, a formula $\varphi_\calC$ in spread normal form with the above property. Recall that a formula is in spread normal form, if it is of the form
      \begin{align*}
      & \quad\bigwedge_{\theta \in Z} \exists x \theta(x) \wedge \forall x \forall y (x = y \vee \nu(x,y)) \\ 
      \wedge &\quad \bigwedge_{k=0}^2 \bigwedge_{h = 1}^{n} \forall x \exists y (\lambda_k(x) \rightarrow (\lambda_{\lfloor k+1 \rfloor}(y) \wedge \mu_h(y)\wedge \theta_h(x,y)))
    \end{align*}
    where (i) $Z$ is a set of unary pure Boolean formulas; (ii) $\nu$, $\theta_1, \ldots, \theta_{n}$ are quantifier- and equality-free formulas; (iii) $\lambda_0$, $\lambda_1$, and $\lambda_2$ are mutually exclusive unary pure Boolean formulas; and (iv) $\mu_1, \ldots, \mu_{n}$ are mutually exclusive unary pure Boolean formulas.

	The formula $\varphi_\calC$ ensures a set of conditions on models. Suppose that $\calC$ has the elements $a_1, \ldots, a_t$ ordered in such a way that $a_1, \ldots, a_s$ are its Kings and $a_{s+1}, \ldots, a_t$ are its non-King elements.

	A first set of conditions, encoded in formulas $\delta_1$ and $\delta_2$, ensures that the court structure of models of $\varphi_\calC$ are isomorphic to $\calC$. To this end, let $Q_1, \ldots, Q_t$ be fresh unary relation symbols.
	
	The formula $\delta_1$ ensures that each such relation $Q_i$ contains exactly one element (intended to be the $i$th element $a_i$ of $\calC$):
	
	\[ \bigwedge_{i=1}^t \exists x Q_i(x) \;\wedge\; \forall x \forall y \Big(x = y \vee \bigwedge_{i=1}^t \neg(Q_i(x) \wedge Q_i(y))\Big)\] 
	
	The formula $\delta_2$ ensures that the substructure induced by elements in the $Q_i$s is isomorphic to $\calC$: 
	
	\[\bigwedge_{i=1}^{t-1}\bigwedge_{j=i+1}^{t} \forall x \forall y \Big(x = y \vee ((Q_i(x) \wedge Q_j(y)) \rightarrow \type^\calC[a_i, a_j])\Big) \]
	
	A second set of conditions, encoded in formulas $\gamma_1$, $\gamma_2$, and $\gamma_3$ ensures that witnesses are spread. To this end, let $Q_{h, 1}, \ldots, Q_{h, s}$ be fresh unary relation symbols. The intuition is that $Q_{h, i}$ shall contain all elements $a$ such that the $\psi_h$-witness of $a$ is the $i$th King. For spreading non-King witnesses, further fresh unary relation symbols $W_1,\ldots, W_{m}$ and $N_0, N_1, N_2$ are used. 
	
	Witnesses for an element $a$ will then be ensured by distinguishing between King witnesses, handled by $\gamma_2$, and non-King witnesses, handeled by $\gamma_3$.
	
	The formula $\gamma_1$ enforces that each element is either a King, or occurs in exactly one $W_i$ and in exactly one $N_i$:
  
  \[\forall x \forall y \Big(x = y \vee \bigvee_{i=1}^s Q_i(x) \vee \mu(x)\Big) \]
  
  Here, $\mu(x)$ is a formula stating that $x$ occurs in exactly one $W_i$ and in exactly one $N_i$. (Note that the quantification of $y$ is not really necessary, but makes it easy to see that $\gamma_2$ adheres to the conditions for spread normal form.)

	The formula $\gamma_2$ ensures that King witnesses are present:
	
	\[\bigwedge_{i=1}^s \bigwedge_{h=1}^m \forall x \forall y \Big(x = y \vee ((Q_{h,i}(x) \wedge Q_i(y)) \rightarrow \psi_h(x, y))\Big)\]

  Finally, the formula $\gamma_3$ ensures that each element has its non-King witnesses:
  
      \[\bigwedge^2_{k=0} \bigwedge^{m}_{h=1}  \forall x \exists y \Big(N_k(x) \to (N_{\lfloor k + 1 \rfloor}(y) \wedge W_h(y) \wedge  \psi^*_h(x, y) )\Big)\]
  where $\psi^*_h(x, y)$ is of the form
    \[\psi^*_h(x, y) \df \Big(\bigwedge_{i=1}^s \neg Q_{h, i}(x)\Big) \to \psi_h(x,y)\]

	Now,  formula $\varphi_\calC$ can be obtained by transforming the conjunction of $\delta_1$, $\delta_2$, $\gamma_1$, $\gamma_2$, $\gamma_3$, and $\nu \df \forall x \forall y (x = y \vee \psi(x,y))$ into spread normal form. Due to the form of these formulas, this can be easily achieved. This concludes the construction of $\varphi_\calC$

	It remains to show that $\varphi$ has a (finite) model with court structure $\calC$ if and only if $\varphi_\calC$ has a (finite) model.
	
	From a model $\calA_\calC$ of $\varphi_\calC$ one can construct a model $\calA$ of $\varphi$ by ignoring the interpretations of all new relation symbols in $\varphi_\calC$, that is, by ignoring the relations $Q^{\calA_\calC}_i$, $Q^{\calA_\calC}_{h, i}$, $W^{\calA_\calC}_i$ and~$N^{\calA_\calC}_i$. The thus obtained structure $\calA$ satisfies the $\forall\forall$-constraints of $\varphi$ due to $\nu$, and the $\forall\exists$-constraints due to $\delta_2$, $\gamma_2$, and $\gamma_3$. The court structure induced by $\calA$ is $\calC$ due to $\delta_2$.

  Conversely suppose that $\calA$ is a model of $\varphi$ with induced court structure $\calC$. A model $\calA_\calC$ of $\varphi_\calC$ can be constructed from $\calA$ as follows. Let $A$ be the universe of $\calA$,  let $A_0 \df \{a_1, \ldots, a_s\}$ be its set of Kings, and let $B \df A \setminus A_0$.

	The structure $\calA_\calC$ is constructed by cloning $B$. It has universe $A_0 \cup \bigcup^2_{k = 0} \bigcup^{m}_{h=1} B_{k, h}$ where each $B_{k, h}$ is a copy of~$B$. 
  
  The relations $Q^{\calA_\calC}_i$ and  $Q^{\calA_\calC}_{h,i}$ are defined as:
  
  \[ Q^{\calA_\calC}_i \df \{ a_i\} \quad \text{and} \quad Q^{\calA_\calC}_{h,i} \df \{b' \mid \text{$b'$ is a copy of an element $b$ with $\psi_h(b, a_i)$}\}\]
  
  The relations $W^{\calA_\calC}_h$ and $N^{\calA_\calC}_k$ are defined as:
    \[W^{\calA_\calC}_h \df \bigcup^2_{k=0} B_{k, h} \quad \text{ and } \quad N^{\calA_\calC}_k \df \bigcup^{m}_{h=1} B_{k, h}\]
  
  It remains to argue that the thus constructed structure $\calA_\calC$ is a model of $\varphi_\calC$. The formulas $\delta_1$ and $\delta_2$ are satisfied, as $\calC$ is the court structure induced by $\calA$. The formula $\gamma_1$ is satisfied due to the construction of $W^{\calA_\calC}_h$ and $N^{\calA_\calC}_k$. The formula $\gamma_2$ is satisfied due the construction of $Q^{\calA_\calC}_{h,i}$. The formula $\gamma_3$ is satisfied due to the construction of $W^{\calA_\calC}_h$ and $N^{\calA_\calC}_k$, and because the $\forall\exists$-constraints of $\varphi$ are satisfied in $\calA$. For seeing that $\nu$ is satisfed we observe that it is satisfied in $\calA$, and cloning does not lead to new binary types (see (i) in the definition of cloning).
\end{proof}

The proofs of Lemma 4.3, Lemma 4.4, and Theorem 4.5  in \cite{Pratt-Hartmann18} implicitly imply the following lemma. 

\begin{lemma}
\label{lemma:spreadImpliesReduction}
 Suppose $\calK$ is a class of structures with signature $\Delta$ and $\varphi$ is a $\fotwo(\calK)$ formula over $\Delta  \uplus \Theta$. If $\varphi$ is in spread normal form then one can effectively construct an equisatisfiable $\fotwo(\calK)$ formula over signature  $\Delta \uplus \Theta'$ where $\Theta'$ only contains unary relation symbols.
\end{lemma}
\begin{proof}We provide two ways of proving the statement. 

  \paragraph{\emph{Adaption of the proof of Pratt-Hartmann \cite{Pratt-Hartmann18}.}} The proofs of Lemma 4.3 and Lemma 4.4 in \cite{Pratt-Hartmann18} are independent of the class $\calK$ of structures. Thus, in particular, Lemma 4.4 can be restated for arbitrary classes $\calK$ with a spread normal form.
  
  Suppose that $\varphi$ is an $\fotwo(\calK)$ formula in spread normal form with signature $\Delta \uplus \Theta$. By Lemma 4.4 in  \cite{Pratt-Hartmann18}, $\varphi$ can be translated into $\varphi^\circ$ such that (1) $\varphi^\circ$ does not use atoms of the form $r(x, y)$ or $r(y, x)$ for $r \in \Theta$, (2) the formula $\varphi$ implies $\varphi^\circ$, and (3) if $\varphi^\circ$ has a model over domain $A$, then so has $\varphi$. 
As the only binary atoms from $\Theta$ occurring in $\varphi^\circ$ are of the form $r(x,x)$ or $r(y,y)$, they can be replaced by new unary atoms, see Theorem~4.5 in \cite{Pratt-Hartmann18}. This yields the intended equisatisfiable formula.
  
  \paragraph{\emph{Direct proof. }}
   Let $\varphi$ be a $\fotwo(\calK)$ formula over signature $\Delta \uplus \Theta$  in spread normal form
    \begin{align*}
      & \quad\bigwedge_{\theta \in Z} \exists x \theta(x) \wedge \forall x \forall y (x = y \vee \nu(x,y)) \\ 
      \wedge &\quad \bigwedge_{k=0}^2 \bigwedge_{h = 0}^{m-1} \forall x \exists y (\lambda_k(x) \rightarrow (\lambda_{\lfloor k+1 \rfloor}(y) \wedge \mu_h(y)\wedge \theta_h(x,y)))
    \end{align*}
    where (i) $Z$ is a set of unary pure Boolean formulas; (ii) $\nu$, $\theta_0, \ldots, \theta_{m-1}$ are quantifier- and equality-free formulas; and (iii) $\lambda_0$, $\lambda_1$, and $\lambda_2$ are mutually exclusive unary pure Boolean formulas; and (iv) $\mu_0, \ldots, \mu_{m-1}$ are mutually exclusive unary pure Boolean formulas.
   
    Suppose, without loss of generality, that the formulas $\nu$ and $\theta_h$ are of the form
      \[\bigvee_i \alpha_i(x,y)\]
    where all $\alpha_i$ are binary types. 
    
    For a binary type $\alpha$, denote by $\alpha^*$ the set of formulas obtained from $\alpha$ by replacing all literals of the form $R(x, y)$, $R(y, x)$,  $\neg R(x, y)$, and $\neg R(y, x)$ for $R \in \Theta$ by \emph{true}. Further denote by $\nu^*$ and $\theta^*_h$ the formulas obtained by replacing all binary types $\alpha$ by $\alpha^*$ in $\nu$ and $\theta_h$.
    
    Consider the formula $\varphi^*$ defined as 
    \begin{align*}
      & \quad\bigwedge_{\theta \in Z} \exists x \theta(x) \wedge \forall x \forall y (x = y \vee \nu^*(x,y)) \\ 
      \wedge &\quad \bigwedge_{k=0}^2 \bigwedge_{h = 0}^{m-1} \forall x \exists y (\lambda_k(x) \rightarrow (\lambda_{\lfloor k+1 \rfloor}(y) \wedge \mu_h(y)\wedge \theta^*_h(x,y)))
    \end{align*}
    This formula only uses literals from $\Delta$, unary literals over $\Theta$ and binary literals of the form $R(x,x), R(y, y), \neg R(x, x), \neg R(y, y)$ over $\Theta$.
    
    We claim that $\varphi^*$ and $\varphi$ are equisatisfiable. Certainly a model $\calA$ of $\varphi$ is also a model of~$\varphi^*$ (as can be seen by inspecting the conjuncts of $\varphi^*$).
    
    For the converse suppose that $\calA^*$ is a model of $\varphi^*$. Then a model $\calA$ of $\varphi$ can be obtained as follows (cf. proof of Lemma 4.4 in \cite{Pratt-Hartmann18}). The structure $\calA$ has the same elements as~$\calA^*$. Further $\calA$ inherits the unary types from $\calA^*$. It remains to fix the binary types for $\calA$.
    
    Fix an element $a$. If $a$ satisfies one of the $\lambda_i$ (and since these are disjoint, no other~$\lambda_j$), then for all $h$ there exists $b$ such that $\calA \models \lambda_{i+1}(b) \wedge \mu_h(b)$. Let $\alpha^*$ be a disjunct of $\theta^*_h$ that is satisfied by $(a, b)$. Suppose $\alpha^*$ was obtained from a type $\beta$ of $\theta_h$ in the construction of~$\varphi^*$, then chose $\beta$ as the binary type of $(a, b)$. Note that $\beta$ is consistent with the (already assigned) unary types of $a$ and $b$, as $\alpha^*$ was obtained from $\beta$ by replacing strictly binary literals by true. In this way, all witnesses for $a$ are fixed. Since all $\mu_h$ are disjoint (since $\varphi$ is in spread normal form), all such witnesses are distinct. 
    
    The above procedure is repeated for all $a$. Note that for a tuple $(a,b)$ the binary types assigned for determining the witnesses of $a$ and $b$ do not conflict since witnesses satisfy $\lambda_{\lfloor i+1 \rfloor}$ which is disjoint from $\lambda_i$. Furthermore, the types assigned in this fashion do not conflict with $\nu$, because all types $\alpha$ used in the formulas $\theta_h$ must be consistent with $\nu$.
    
    It remains to assign the binary types in $\calA$ for all remaining tuples $(a, b)$. This has to be done in accordance with the condition~$\nu$. As $\calA^*$ satisfies $\varphi^*$, the tuple $(a, b)$ satisfies $\nu^*$ and hence some formula $\alpha^*$ that was obtained from some binary type $\beta$ in $\nu$. The binary type of $(a, b)$ is chosen as $\beta$. This, again, does not conflict with previously chosen unary types of $a$ and $b$.
    
    Hence $\varphi$ and $\varphi^*$ are equisatisfiable. Now, an equisatisfiable formula $\psi$ for $\varphi$  over signature  $\Delta \uplus \Theta'$ where $\Theta'$ only contains unary relation symbols can be obtained from $\varphi^*$ as follows (cf. proof of Theorem 4.5 in \cite{Pratt-Hartmann18}). All binary atoms in $\varphi^*$ with relation symbols from $\Theta$ are of the form $R(x,x)$ or $R(y,y)$. Replacing these atoms by atoms $\widehat{R}(x)$ and $\widehat{R}(y)$ where $\widehat{R}$ is a fresh unary relation symbol yields the equisatisfiable formula $\psi$ over signature $\Delta \uplus \Theta'$ where \[\Theta' \df \{ S \in \Theta \mid S \text{ is unary}\} \uplus \{ \widehat{R} \in \Theta \mid R \in \Theta \text{ is binary}\}.
      \tag*{\qedhere}
    \]
\end{proof}

\begin{proof}[Proof of Proposition \ref{proposition:binaryToUnary}]
By Lemma \ref{lemma:cloningOfTwoLinearOrders} structures with two linear order allow cloning, and therefore each $\fotwo(<_1, <_2)$ formula can be transformed into an equisatisfiable disjunction $\bigvee_i \varphi_i$ of formulas $\varphi_i$ over signature $\{<_1, <_2\} \uplus \Theta_i$ in spread normal form by  Lemma \ref{lemma:cloningImpliesSpreading}. Now, each $\varphi_i$ can be transformed into an equisatisfiable formula $\varphi'_i$ over $\{<_1, <_2\} \uplus \Theta'_i$ where each $\Theta'_i$ is unary, by  Lemma \ref{lemma:spreadImpliesReduction}. Now, Proposition \ref{proposition:binaryToUnary} follows from the fact that $\bigvee_i \varphi'_i$ is equisatisfiable to~$\bigvee_i \varphi_i$. 
\end{proof}
Recall that Theorem~\ref{theorem:maintheorem} 
implies the (general) decidability of $\emsotwo(<_1, <_2)$, see Corollary~\ref{corollary:emsotwo:twolinearorders}.
Therefore, Proposition~\ref{proposition:binaryToUnary} yields Theorem~\ref{thm:esotwo:twolinearorders}.

\section*{Acknowledgment}
	We are grateful to Nathan Lhote for insightful discussions, in particular on the inclusion of MSO-definable atoms in Theorem \ref{theorem:maintheorem}. We also thank the anonymous reviewers for helpful suggestions on improving the readability of the article.

\bibliographystyle{alphaurl}
\bibliography{bibliography}

\begin{thebibliography}{BMSS09}

\bibitem[BCL08]{BlumensathCL08}
Achim Blumensath, Thomas Colcombet, and Christof L{\"{o}}ding.
\newblock Logical theories and compatible operations.
\newblock In {\em Logic and Automata: History and Perspectives [in Honor of
  Wolfgang Thomas]}, pages 73--106, 2008.

\bibitem[BKL08]{Baier08}
Christel Baier, Joost-Pieter Katoen, and Kim~Guldstrand Larsen.
\newblock {\em Principles of model checking}.
\newblock MIT press, 2008.

\bibitem[BMSS09]{BojanczykMSS09}
Miko{\l}aj Boja{\'{n}}czyk, Anca Muscholl, Thomas Schwentick, and Luc Segoufin.
\newblock Two-variable logic on data trees and {XML} reasoning.
\newblock {\em J. {ACM}}, 56(3):13:1--13:48, 2009.
\newblock \href {https://doi.org/10.1145/1516512.1516515}
  {\path{doi:10.1145/1516512.1516515}}.

\bibitem[Boj09]{DBLP:conf/dlt/Bojanczyk09}
Mikolaj Bojanczyk.
\newblock Factorization forests.
\newblock In Volker Diekert and Dirk Nowotka, editors, {\em Developments in
  Language Theory, 13th International Conference, {DLT} 2009, Stuttgart,
  Germany, June 30 - July 3, 2009. Proceedings}, volume 5583 of {\em Lecture
  Notes in Computer Science}, pages 1--17. Springer, 2009.
\newblock \href {https://doi.org/10.1007/978-3-642-02737-6\_1}
  {\path{doi:10.1007/978-3-642-02737-6\_1}}.

\bibitem[CW13]{CharatonikW13}
Witold Charatonik and Piotr Witkowski.
\newblock Two-variable logic with counting and trees.
\newblock In {\em {LICS}'13}, pages 73--82, 2013.

\bibitem[DFL18]{DartoisFL18}
Luc Dartois, Emmanuel Filiot, and Nathan Lhote.
\newblock Logics for word transductions with synthesis.
\newblock In Anuj Dawar and Erich Gr{\"{a}}del, editors, {\em Proceedings of
  the 33rd Annual {ACM/IEEE} Symposium on Logic in Computer Science, {LICS}
  2018, Oxford, UK, July 09-12, 2018}, pages 295--304. {ACM}, 2018.
\newblock \href {https://doi.org/10.1145/3209108.3209181}
  {\path{doi:10.1145/3209108.3209181}}.

\bibitem[GKV97]{GradelKV97}
Erich Gr{\"a}del, Phokion~G. Kolaitis, and Moshe~Y. Vardi.
\newblock On the decision problem for two-variable first-order logic.
\newblock {\em Bulletin of Symbolic Logic}, 3(1):53--69, 1997.

\bibitem[GO99]{GraedelO1999}
Erich Gr{\"{a}}del and Martin Otto.
\newblock On logics with two variables.
\newblock {\em Theor. Comput. Sci.}, 224(1-2):73--113, 1999.

\bibitem[JLZ20]{JungLZ20}
Jean~Christoph Jung, Carsten Lutz, and Thomas Zeume.
\newblock On the decidability of expressive description logics with transitive
  closure and regular role expressions.
\newblock In Diego Calvanese, Esra Erdem, and Michael Thielscher, editors, {\em
  Proceedings of the 17th International Conference on Principles of Knowledge
  Representation and Reasoning, {KR} 2020, Rhodes, Greece, September 12-18,
  2020}, pages 529--538, 2020.
\newblock \href {https://doi.org/10.24963/kr.2020/53}
  {\path{doi:10.24963/kr.2020/53}}.

\bibitem[KF94]{KaminskiF94}
Michael Kaminski and Nissim Francez.
\newblock Finite-memory automata.
\newblock {\em Theor. Comput. Sci.}, 134(2):329--363, 1994.
\newblock \href {https://doi.org/10.1016/0304-3975(94)90242-9}
  {\path{doi:10.1016/0304-3975(94)90242-9}}.

\bibitem[Kie11]{Kieronski11}
Emanuel Kieronski.
\newblock Decidability issues for two-variable logics with several linear
  orders.
\newblock In {\em {CSL} 2011}, volume~12 of {\em LIPIcs}, pages 337--351, 2011.

\bibitem[KO12]{KieronskiO12}
Emanuel Kieronski and Martin Otto.
\newblock Small substructures and decidability issues for first-order logic
  with two variables.
\newblock {\em J. Symb. Log.}, 77(3):729--765, 2012.

\bibitem[KT09]{KieronskiT09}
Emanuel Kieronski and Lidia Tendera.
\newblock On finite satisfiability of two-variable first-order logic with
  equivalence relations.
\newblock In {\em LICS 2009}, pages 123--132. {IEEE} Computer Society, 2009.

\bibitem[Lib04]{Libkin04}
Leonid Libkin.
\newblock {\em Elements of Finite Model Theory}.
\newblock Springer, 2004.

\bibitem[Mor75]{Mortimer75}
Michael Mortimer.
\newblock On languages with two variables.
\newblock {\em Zeitschr. f. math. Logik u. Grundlagen d. Math.}, 21:135--140,
  1975.

\bibitem[Nev02]{Neven02}
Frank Neven.
\newblock Automata theory for {XML} researchers.
\newblock {\em ACM Sigmod Record}, 31(3):39--46, 2002.

\bibitem[NSV04]{NevenSV04}
Frank Neven, Thomas Schwentick, and Victor Vianu.
\newblock Finite state machines for strings over infinite alphabets.
\newblock {\em {ACM} Trans. Comput. Log.}, 5(3):403--435, 2004.
\newblock \href {https://doi.org/10.1145/1013560.1013562}
  {\path{doi:10.1145/1013560.1013562}}.

\bibitem[Ott01]{Otto01}
Martin Otto.
\newblock Two variable first-order logic over ordered domains.
\newblock {\em J. Symb. Log.}, 66(2):685--702, 2001.

\bibitem[Pra18]{Pratt-Hartmann18}
Ian Pratt{-}Hartmann.
\newblock Finite satisfiability for two-variable, first-order logic with one
  transitive relation is decidable.
\newblock {\em Math. Log. Q.}, 64(3):218--248, 2018.
\newblock \href {https://doi.org/10.1002/malq.201700055}
  {\path{doi:10.1002/malq.201700055}}.

\bibitem[Rab72]{Rabin72}
Michael~Oser Rabin.
\newblock {\em Automata on Infinite Objects and Church's Problem}.
\newblock Number~13. American Mathematical Soc., 1972.

\bibitem[Sco62]{Scott1962}
Dana Scott.
\newblock A decision method for validity of sentences in two variables.
\newblock {\em Journal of Symbolic Logic}, 27(377):74, 1962.

\bibitem[She75]{Shelah75}
Saharon Shelah.
\newblock The monadic theory of order.
\newblock {\em Annals of Mathematics}, pages 379--419, 1975.
\newblock \href {https://doi.org/10.2307/1971037} {\path{doi:10.2307/1971037}}.

\bibitem[SZ12]{SchwentickZ12}
Thomas Schwentick and Thomas Zeume.
\newblock Two-variable logic with two order relations.
\newblock {\em Logical Methods in Computer Science}, 8(1), 2012.

\bibitem[Tho97]{Thomas97}
Wolfgang Thomas.
\newblock Languages, automata, and logic.
\newblock In {\em Handbook of formal languages}, pages 389--455. Springer,
  1997.

\bibitem[TZ20]{TorunczykZ20}
Szymon Torunczyk and Thomas Zeume.
\newblock Register automata with extrema constraints, and an application to
  two-variable logic.
\newblock In Holger Hermanns, Lijun Zhang, Naoki Kobayashi, and Dale Miller,
  editors, {\em {LICS} '20: 35th Annual {ACM/IEEE} Symposium on Logic in
  Computer Science, Saarbr{\"{u}}cken, Germany, July 8-11, 2020}, pages
  873--885. {ACM}, 2020.
\newblock \href {https://doi.org/10.1145/3373718.3394748}
  {\path{doi:10.1145/3373718.3394748}}.

\bibitem[ZH16]{ZeumeH16}
Thomas Zeume and Frederik Harwath.
\newblock Order-invariance of two-variable logic is decidable.
\newblock In Martin Grohe, Eric Koskinen, and Natarajan Shankar, editors, {\em
  Proceedings of the 31st Annual {ACM/IEEE} Symposium on Logic in Computer
  Science, {LICS} '16, New York, NY, USA, July 5-8, 2016}, pages 807--816.
  {ACM}, 2016.
\newblock \href {https://doi.org/10.1145/2933575.2933594}
  {\path{doi:10.1145/2933575.2933594}}.

\end{thebibliography}

\end{document}